\documentclass[journal]{IEEEtran}

\usepackage[cmex10]{amsmath}
\usepackage{amssymb,amsthm}
\usepackage{color}

\theoremstyle{definition}
\newtheorem{definition}{Definition}
\newtheorem{example}{Example}

\theoremstyle{plain}
\newtheorem{theorem}{Theorem}
\newtheorem{proposition}{Proposition}
\newtheorem{lemma}{Lemma}
\newtheorem{remark}{Remark}
\newtheorem{corollary}{Corollary}

\begin{document}
%
\title{Generalized rank weights of reducible codes, optimal cases and related properties}
%
%
%

\author{Umberto Mart{\'i}nez-Pe\~{n}as,~\IEEEmembership{Student Member,~IEEE,} 
\thanks{This work was supported by The Danish Council for Independent Research under Grant No. DFF-4002-00367 and Grant No. DFF-5137-00076B (EliteForsk-Rejsestipendium).

Parts of this paper were presented at the IEEE International Symposium on Information Theory, Barcelona, Spain, Jul 2016. \cite{grwreducible}

U. Mart{\'i}nez-Pe\~{n}as is with the Department of Mathematical Sciences, Aalborg University, Aalborg 9220, Denmark (e-mail: umberto@math.aau.dk).}}

%
%

\markboth{}%
{Shell \MakeLowercase{\textit{et al.}}: Bare Demo of IEEEtran.cls for Journals}
%



\maketitle

\begin{abstract}
Reducible codes for the rank metric were introduced for cryptographic purposes. They have fast encoding and decoding algorithms, include maximum rank distance (MRD) codes and can correct many rank errors beyond half of their minimum rank distance, which makes them suitable for error-correction in network coding. In this paper, we study their security behaviour against information leakage on networks when applied as coset coding schemes, giving the following main results: 1) we give lower and upper bounds on their generalized rank weights (GRWs), which measure worst-case information leakage to the wire-tapper, 2) we find new parameters for which these codes are MRD (meaning that their first GRW is optimal), and use the previous bounds to estimate their higher GRWs, 3) we show that all linear (over the extension field) codes whose GRWs are all optimal for fixed packet and code sizes but varying length are reducible codes up to rank equivalence, and 4) we show that the information leaked to a wire-tapper when using reducible codes is often much less than the worst case given by their (optimal in some cases) GRWs. We conclude with some secondary related properties: Conditions to be rank equivalent to cartesian products of linear codes, conditions to be rank degenerate, duality properties and MRD ranks. \\
\end{abstract}

\begin{IEEEkeywords}
Generalized rank weight, rank-metric codes, rank distance, rank equivalent codes, reducible codes, secure network coding.
\end{IEEEkeywords}

%
\IEEEpeerreviewmaketitle

\section{Introduction}
%
%
%
%

\IEEEPARstart{L}{inear} network coding was first studied in \cite{ahlswede, linearnetwork}, further formalized in \cite{Koetter2003}, and provides higher throughput than storing and forwarding messages on the network. Two of the \textit{main problems} in this context are error and erasure correction, and security against information leakage to a wire-tapper, which were first studied in \cite{cai-yeung} and \cite{secure-network}, respectively. 

Rank-metric codes were found to be universally suitable (meaning independently of the underlying network code) for error and erasure correction in linear network coding in \cite{on-metrics}, used as forward error-correcting codes, and they were found to be universally suitable against information leakage in \cite{silva-universal}, used in the form of coset coding. Both constructions can be treated separately and applied together in a concatenated way (see \cite[Sec. VII-B]{silva-universal}). 

On the security side, generalized rank weights (GRWs) of codes that are linear over the extension field were introduced in \cite{rgrw, oggier} to measure the worst-case information leakage for a given number of wire-tapped links. Later, GRWs were extended in \cite{ravagnaniweights} and \cite{allertonversion} to codes that are linear over the base field, where they are called Delsarte generalized weights and generalized matrix weights, respectively. We will use the term GRWs for the latter parameters, which were also found to measure the worst-case information leakage for codes that are linear over the base field \cite[Th. 3]{allertonversion}.

Gabidulin codes \cite{gabidulin} constitute a family of maximum rank distance (MRD) codes that cover all cases when the number of outgoing links $ n $ is not larger than the packet length $ m $, and all of their GRWs are optimal (meaning largest possible). 

Cartesian products of these codes are proposed in \cite[Sec. VII.C]{silva-universal} for the case $ n > m $ both for error correction and security against information leakage. A generalization of these codes, called reducible codes, were introduced earlier in \cite{reducible} as an alternative to Gabidulin codes \cite{gabidulin} to improve the security of rank-based public key cryptosystems \cite{ideals}. On the error correction side, it was shown in \cite{reducible} that reducible codes have fast encoding and rank error-correcting algorithms, their minimum rank distance is not worse than that of cartesian products of codes \cite[Sec. VII.C]{silva-universal}, being actually MRD in some cases, and they can correct many rank errors beyond half of their minimum rank distance (even in the MRD cases). Therefore they seem to be the best known codes for error correction in linear network coding when $ n > m $. 

However, on the security side, only the existence of codes with optimal first GRW (MRD codes) has been studied in the case $ n > m $ \cite[Sec. IV-B]{allertonversion}, but no bounds nor estimates of higher GRWs of rank-metric codes or other properties related to their worst-case information leakage are known when $ n > m $, except for cyclic codes with minimal GRWs \cite{oggiercyclic}.

In this paper, we study the security provided by reducible codes in linear network coding when used for coset coding as in \cite{silva-universal} by studying their GRWs and showing their optimality in several cases. In particular, we study for the first time the GRWs of a concrete family of rank-metric codes with $ n > m $, which moreover include MRD codes for several parameters.

\subsection{Main contributions}

Our main contributions are the following:
\begin{enumerate}
\item
We give lower and upper bounds on GRWs of reducible codes, and exact values for cartesian products, giving a first step in the open problem of estimating or bounding the GRWs of a family of rank-metric codes for $ n > m $. 
\item
We give new families of parameters for which reducible codes are MRD (some were given in \cite{reducible}), meaning that their first GRW is optimal and thus they are optimal regarding zero information leakage among all linear (over the extension or the base field) codes, by \cite[Th. 3]{allertonversion}. Using the estimates and exact values of GRWs of these codes in the previous item, we also give a first step in the open problem of finding the GRWs of a family of MRD codes for $ n > m $.
\item  
We show that all linear (over the extension field) codes whose GRWs are all optimal for fixed packet and code sizes, but varying length, lie in the family of reducible codes from the previous item, up to rank equivalence. 
\item
Finally, we show that information leakage when using reducible codes is often much less than the worst case given by their GRWs. In particular, they often provide strictly higher security than the known security provided by other MRD codes \cite[Sec. IV-B]{allertonversion}.
\end{enumerate}

\subsection{Organization of the paper}

After some preliminaries in Section II, the paper is organized as follows: In Section III, we give lower and upper bounds on the GRWs of reducible codes, extending the lower bound on the minimum rank distance given in \cite{reducible}, and see that the given upper bound on the minimum rank distance can be reached by some reduction. In Section IV, we obtain new parameters for which reducible codes are MRD (or close to MRD) and with MRD components, and obtain explicit estimates on their GRWs, including those MRD codes found in \cite{reducible} and considered for secure network coding in \cite{silva-universal}. In Section V, we obtain all linear codes whose GRWs are all optimal, for all fixed packet and code sizes, up to rank equivalence. In Section VI, we see that the actual information leakage occuring when using reducible codes is often much less than the worst case given by their GRWs, providing higher security than other known MRD codes. Finally, in Section VII, we study secondary but related properties: Conditions to be rank equivalent to cartesian products and conditions to be rank degenerate. We study their duality properties and MRD ranks. Finally, we propose alternative constructions to the classical $ (\mathbf{u}, \mathbf{u} + \mathbf{v}) $ construction.

\section{Definitions and preliminaries}

\subsection{Rank-metric codes}

Fix a prime power $ q $ and positive integers $ m $ and $ n $, and let $ \mathbb{F}_q $ and $ \mathbb{F}_{q^m} $ denote the finite fields with $ q $ and $ q^m $ elements, respectively. We may identify vectors in $ \mathbb{F}_{q^m}^n $ with $ m \times n $ matrices over $ \mathbb{F}_q $: Fix a basis $ \alpha_1, \alpha_2, \ldots, \alpha_m $ of $ \mathbb{F}_{q^m} $ over $ \mathbb{F}_q $. If $ \mathbf{c} = (c_1, c_2, \ldots, c_n) \in \mathbb{F}_{q^m}^n $, $ c_j = \sum_{i = 1}^m \alpha_i c_{i,j} $, and $ c_{i,j} \in \mathbb{F}_q $, for $ i = 1,2, \ldots, m $ and $ j = 1,2, \ldots, n $, we may identify $ \mathbf{c} $ with the matrix
\begin{equation}
M(\mathbf{c}) = \left( c_{i,j} \right)_{1 \leq j \leq n}^{1 \leq i \leq m}.
\label{eq def matrix representation}
\end{equation}
The rank weight of a vector $ \mathbf{c} \in \mathbb{F}_{q^m}^n $ is defined as the rank of the matrix $ M(\mathbf{c}) $ and denoted by $ {\rm wt_R}(\mathbf{c}) $. In this paper, a code is a subset of $ \mathbb{F}_{q^m}^n $. The term \textit{rank-metric code} is used for codes with the rank metric.

\subsection{Universal secure linear network coding} \label{subsec universal secure}

We consider a network with several sources and several sinks as in \cite{ahlswede, linearnetwork}. In this model, a given source wants to transmit $ k $ packets in $ \mathbb{F}_q^m $ to one or several sink nodes, and does so by encoding them into a vector, $ \mathbf{c} \in \mathbb{F}_{q^m}^n $, which can be seen as $ n $ packets in $ \mathbb{F}_q^m $ by (\ref{eq def matrix representation}), being $ n $ the number of outgoing links from the source.

In \textit{linear network coding}, as considered in \cite{ahlswede} and \cite{linearnetwork}, the nodes in the network forward linear combinations of received packets (see \cite[Definition 1]{Koetter2003}), achieving higher throughput than just storing and forwarding. This means that a given sink is assumed to receive the vector
$$ \mathbf{y} = \mathbf{c} A^T \in \mathbb{F}_{q^m}^N, $$
for some matrix $ A \in \mathbb{F}_q^{N \times n} $, called a transfer matrix. 

Two of the \textit{main problems} in linear network coding considered in the literature are the following: 
\begin{enumerate}
\item
Error correction \cite{cai-yeung}: Several packets are injected on some links in the network, hence the sink receives 
$$ \mathbf{y} = \mathbf{c} A^T + \mathbf{e} \in \mathbb{F}_{q^m}^N, $$
for an error vector $ \mathbf{e} \in \mathbb{F}_{q^m}^N $. 
\item
Information leakage \cite{secure-network}: A wire-tapper listens to $ \mu > 0 $ links in the network, obtaining 
$$ \mathbf{z} = \mathbf{c} B^T \in \mathbb{F}_{q^m}^{\mu}, $$
for a matrix $ B \in \mathbb{F}_q^{\mu \times n} $. 
\end{enumerate}

In \cite{on-metrics}, it is proven that rank-metric codes are suitable for error correction when used as forward error-correcting codes, and in \cite{silva-universal}, it is proven that they are also suitable to protect messages from information leakage when used as coset coding schemes, which were introduced in \cite{wyner} and \cite{ozarow}. Both coding techniques can be treated separately and applied together in a concatenated way (see \cite[Sec. VII-B]{silva-universal}). 

Moreover, rank-metric codes are \textit{universal} \cite{silva-universal} in the sense that they correct a given number of errors and erasures, and protect against a given number of wire-tapped links, independently of the matrices $ A $ and $ B $, respectively. 

We consider the particular coding schemes in \cite[Sec. V-B]{silva-universal} with uniform distributions:

\begin{definition} [\textbf{\cite{silva-universal}}] \label{definition ozarow}
Given an $ \mathbb{F}_{q^m} $-linear code $ C \subseteq \mathbb{F}_{q^m}^n $ with generator matrix $ G \in \mathbb{F}_{q^m}^{k \times n} $, we define its coset coding scheme as follows: For $ \mathbf{x} \in \mathbb{F}_{q^m}^k $, its coset encoding is a vector $ \mathbf{c} \in \mathbb{F}_{q^m}^n $ chosen uniformly at random and such that $ \mathbf{x} = \mathbf{c} G^T $. 
\end{definition}

This type of encoding has been recently extended to $ \mathbb{F}_q $-linear codes in \cite[Sec. II-D]{allertonversion}.

In this paper we will focus on rank-metric codes used for security against information leakage in the form of coset coding.

\subsection{Generalized rank weights and information leakage} \label{subsec GRWs and leakage}

The information leaked to a wire-tapping adversary when using coset coding schemes was obtained in \cite[Lemma 6]{silva-universal}, then generalized in \cite[Lemma 7]{rgrw} to $ \mathbb{F}_{q^m} $-linear nested coset coding schemes \cite{zamir}, and in \cite[Prop. 4]{allertonversion} to $ \mathbb{F}_q $-linear coset coding schemes. 

We need the concept of Galois closed spaces \cite{stichtenoth}:

\begin{definition}[\textbf{\cite{stichtenoth}}]
Denote $ [i] = q^i $ for an integer $ i \geq 0 $. If $ C \subseteq \mathbb{F}_{q^m}^n $ is $ \mathbb{F}_{q^m} $-linear, we denote 
$$ C^{[i]} = \{ (c_1^{[i]}, c_2^{[i]}, \ldots, c_n^{[i]}) \mid (c_1, c_2, \ldots, c_n) \in C \}, $$
we define the Galois closure of $ C $ as $ C^* = \sum_{i=0}^{m-1} C^{[i]} $, and we say that it is Galois closed if $ C = C^* $.
\end{definition}

The next lemma is \cite[Lemma 6]{silva-universal}. Throughout the paper, $ I(X;Y) $ denotes the mutual information of the random variables $ X $ and $ Y $, taking logarithms with base $ q^m $.

\begin{lemma}[\textbf{\cite{silva-universal}}] \label{lemma silva}
Denote by $ S $ the uniform random variable in $ \mathbb{F}_{q^m}^k $, $ X $ its coset encoding using an $ \mathbb{F}_{q^m} $-linear code $ C \subseteq \mathbb{F}_{q^m}^n $ according to Definition \ref{definition ozarow}, and denote $ W = XB^T $, where $ B \in \mathbb{F}_q^{\mu \times n} $. Then
\begin{equation} \label{leaked info}
I(S;W) = \dim(C \cap V),
\end{equation}
where $ V \subseteq \mathbb{F}_{q^m}^n $ is the $ \mathbb{F}_{q^m} $-linear vector space with generator matrix $ B $.
\end{lemma}

Since Galois closed spaces in $ \mathbb{F}_{q^m}^n $ are those $ \mathbb{F}_{q^m} $-linear spaces with a generator matrix over $ \mathbb{F}_q $ \cite[Lemma 1]{stichtenoth}, the previous lemma motivates the definition of generalized rank weights, introduced independently in \cite{oggier} for $ n \leq m $, and in \cite[Def. 2]{rgrw} for the general case:

\begin{definition}[\textbf{\cite{rgrw}}]
Given an $ \mathbb{F}_{q^m} $-linear code $ C \subseteq \mathbb{F}_{q^m}^n $ of dimension $ k $, we define its $ r $-th generalized rank weight (GRW), for $ 1 \leq r \leq k $, as
\begin{equation*}
\begin{split}
d_{R,r}(C) = \min \{ & \dim (V) \mid V \subseteq \mathbb{F}_{q^m}^n, \mathbb{F}_{q^m} \textrm{-linear and} \\
 & V=V^*, \dim(C \cap V) \geq r \}.
\end{split}
\end{equation*}
We also define $ d_{R,0}(C) = 0 $ for convenience.
\end{definition}

Hence $ d_{R,r}(C) $ is the minimum number of links that a wire-tapper needs to listen to in order to obtain at least the amount of information contained in $ r $ packets. In other words, $ r-1 $ packets is the \textit{worst-case information leakage} when at most $ d_{R,r}(C)-1 $ links are wire-tapped. 

The next lemma corresponds to \cite[Th. 16, Cor. 17]{slides}:

\begin{lemma}[\textbf{\cite{slides}}] \label{equivalent definition slides}
Given an $ \mathbb{F}_{q^m} $-linear code $ C \subseteq \mathbb{F}_{q^m}^n $ of dimension $ k $ and $ 1 \leq r \leq k $, it holds that
\begin{equation*}
\begin{split}
d_{R,r}(C) = \min \{ & {\rm wt_R}(D) \mid D \subseteq C, \mathbb{F}_{q^m} \textrm{-linear and} \\
 & \dim(D) = r \},
\end{split}
\end{equation*}
where $ {\rm wt_R}(D) = \dim(D^*) $ for an $ \mathbb{F}_{q^m} $-linear $ D \subseteq \mathbb{F}_{q^m}^n $.
\end{lemma}

In particular, it is shown in \cite[Cor. 1]{rgrw} that $ d_{R,1}(C) $ is the minimum rank distance of the code $ C $ (also denoted by $ d_R(C) $). Thus the minimum rank distance is of particular importance, since it gives the maximum number of wire-tapped links that guarantee zero information leakage, and we may evaluate the code's optimality among all rank-metric codes (linear and non-linear) in this sense using the Singleton bound \cite[Th. 6.3]{delsartebilinear}:
\begin{equation} \label{singleton bound}
\# C \leq q^{\max \{ m,n \} (\min \{ m,n \} - d_R(C) + 1)},
\end{equation}
where $ C \subseteq \mathbb{F}_{q^m}^n $ is an arbitrary rank-metric code. Codes attaining this bound are called maximum rank distance (MRD) codes.

\subsection{Existing MRD code constructions} \label{existing code constructions}

We briefly revisit two existing code constructions that have already been considered in the literature:

\begin{enumerate}
\item
Assume $ n \leq m $ and $ 1 \leq k \leq n $: Take elements $ \beta_1, \beta_2, \ldots, \beta_n \in \mathbb{F}_{q^m} $ that are linearly independent over $ \mathbb{F}_q $. The $ \mathbb{F}_{q^m} $-linear code $ C_{Gab} \subseteq \mathbb{F}_{q^m}^n $ generated by the matrix
\begin{displaymath}
\left( \begin{array}{cccc}
\beta_1 & \beta_2 & \ldots & \beta_n \\
\beta_1^{[1]} & \beta_2^{[1]} & \ldots & \beta_n^{[1]} \\
\vdots & \vdots & \ddots & \vdots \\
\beta_1^{[k-1]} & \beta_2^{[k-1]} & \ldots & \beta_n^{[k-1]} \\
\end{array} \right)
\end{displaymath}
has dimension $ k $ and minimum rank distance $ d_{R}(C_{Gab}) = n - k + 1 $, and hence is MRD. These codes are known as Gabidulin codes and were introduced in \cite{gabidulin}. Their GRWs were given in \cite[Cor. 2]{rgrw}: 
$$ d_{R,r}(C_{Gab}) = n - k + r. $$
\item
Assume $ n = lm $ and $ k = lk^\prime $, for some positive integers $ l $ and $ k^\prime \leq m $: The $ \mathbb{F}_{q^m} $-linear code $ C \subseteq \mathbb{F}_{q^m}^n $ defined as $ C = C_1 \times C_2 \times \cdots \times C_l $, where each $ C_i \subseteq \mathbb{F}_{q^m}^m $ is a $ k^\prime $-dimensional Gabidulin code, has dimension $ k $ and minimum rank distance $ d_{R}(C) = m - k^\prime + 1 $, and hence is also MRD. These codes were introduced in \cite[Cor. 1]{reducible} and considered in \cite[Sec. VII-C]{silva-universal} for secure network coding. In contrast with Gabidulin codes, although a first analysis of these codes is given in \cite{silva-universal}, their GRWs are still not known. We will find all of them in Section \ref{estimates on their}.
\end{enumerate}

The two previous constructions are particular cases of reducible codes, introduced in \cite{reducible}, which we will study in the rest of the paper.

In \cite{allertonversion}, MRD $ \mathbb{F}_q $-linear codes obtained by transposing the matrix representations of codewords in a Gabidulin code are proposed for the case $ n > m $. However no exact values or lower bounds are known for any code in this case ($ n > m $).

\subsection{Reducible codes and reductions} \label{subsec reducible}

Consider positive integers $ l, n_1, n_2, \ldots, n_l $ and $ \mathbb{F}_{q^m} $-linear codes $ C_1 \subseteq \mathbb{F}_{q^m}^{n_1}, C_2 \subseteq \mathbb{F}_{q^m}^{n_2}, \ldots, C_l \subseteq \mathbb{F}_{q^m}^{n_l} $ of dimensions $ k_1 $, $ k_2, \ldots, k_l $, respectively. Consider matrices $ G_{i,j} \in \mathbb{F}_{q^m}^{k_i \times n_j} $, for $ i =1,2, \ldots, l $ and $ j = i,i+1, \ldots, l $, where $ G_{i,i} $ generates $ C_i $. 

\begin{definition} [\textbf{\cite{reducible}}] \label{def reducible}
We say that an $ \mathbb{F}_{q^m} $-linear code $ C \subseteq \mathbb{F}_{q^m}^n $ is reducible with reduction $ \mathcal{R} = (G_{i,j})_{1 \leq i \leq l}^{i \leq j \leq l} $ if it has a generator matrix of the form
\begin{displaymath}
G = \left(
\begin{array}{cccccc}
G_{1,1} & G_{1,2} & G_{1,3} & \ldots & G_{1, l-1} & G_{1,l} \\
0 & G_{2,2} & G_{2,3} & \ldots & G_{2, l-1} & G_{2,l} \\
0 & 0 & G_{3,3} & \ldots & G_{3, l-1} & G_{3,l} \\
\vdots & \vdots & \vdots & \ddots & \vdots & \vdots \\
0 & 0 & 0 & \ldots & G_{l-1,l-1} & G_{l-1,l} \\
0 & 0 & 0 & \ldots & 0 & G_{l,l} \\
\end{array} \right).
\end{displaymath}
\end{definition}

The length of the code $ C $ is $ n = n_1 + n_2 + \cdots + n_l $ and its dimension is $ k = k_1 + k_2 + \cdots + k_l $. $ C $ is the cartesian product of the codes $ C_1, C_2, \ldots, C_l $ if $ G_{i,j} = 0 $, for all $ j > i $.

\begin{definition}
For a given reduction $ \mathcal{R} $ as in the previous definition, we define its main components as the codes $ C_1, C_2, \ldots, C_l $, its row components as the $ \mathbb{F}_{q^m} $-linear codes $ C_i^\prime \subseteq \mathbb{F}_{q^m}^n $ with generator matrices
\begin{equation} \label{row components}
G_i^\prime = (0, \ldots, 0, G_{i,i}, G_{i,i+1}, \ldots, G_{i,l}),
\end{equation}
for $ i = 1,2, \ldots, l $, and its column components as the $ \mathbb{F}_{q^m} $-linear codes $ \widehat{C}_j \subseteq \mathbb{F}_{q^m}^{n_j} $ generated by the matrices
\begin{equation} \label{column components}
\widehat{G}_j = (G_{1,j}, G_{2,j}, \ldots, G_{j,j})^T,
\end{equation}
for $ j = 1,2, \ldots, l $, which need not have full rank. 
\end{definition}

It holds that $ k_i = \dim(C_i^\prime) $, $ \widehat{k}_j = \dim(\widehat{C}_j) \geq k_j $, $ C = C_1^\prime \oplus C_2^\prime \oplus \cdots \oplus C_l^\prime $ and $ C \subseteq \widehat{C} = \widehat{C}_1 \times \widehat{C}_2 \times \cdots \times \widehat{C}_l $.

Different reductions always have the same main components if their block sizes are the same. See Appendix \ref{app 1} for a discussion on the uniqueness of reductions of a reducible code.

\section{Bounds on GRWs of reducible codes and exact values} \label{sec bounds on GRWs}

With notation as in Subsection \ref{subsec reducible}, it is proven in \cite[Lemma 2]{reducible} that 
\begin{equation} \label{bound min distance}
d_{R,1}(C) \geq \min \{ d_{R,1}(C_1), d_{R,1}(C_2), \ldots, d_{R,1}(C_l) \}. 
\end{equation}
We now present the main result of this section, which generalizes (\ref{bound min distance}) to higher GRWs and also gives upper bounds. As observed below, it gives the exact values for cartesian products.

\begin{theorem} \label{grw bounds}
With notation as in Subsection \ref{subsec reducible}, for every $ r = 1,2, \ldots, k $, we have that
\begin{equation} \label{lower bound}
\begin{split}
d_{R,r}(C) \geq \min \{ & d_{R,r_1}(C_1) + d_{R,r_2}(C_2) + \cdots + d_{R,r_l}(C_l) \\
 & \mid r = r_1 + r_2 + \cdots + r_l, 0 \leq r_i \leq k_i \},
\end{split}
\end{equation} 
and
\begin{equation} \label{upper bound}
\begin{split}
d_{R,r}(C) \leq \min \{ & d_{R,r_1}(C_1^\prime) + d_{R,r_2}(C_2^\prime) + \cdots + d_{R,r_l}(C_l^\prime) \\
 & \mid r = r_1 + r_2 + \cdots + r_l, 0 \leq r_i \leq k_i \}.
\end{split}
\end{equation}
\end{theorem}

The proof can be found at the end of the section. We now elaborate on some particular cases of interest.

First, observe that the bound (\ref{lower bound}) gives the bound (\ref{bound min distance}) for the minimum rank distance (the case $ r=1 $), and the bound (\ref{upper bound}) gives the following (immediate) upper bound:
\begin{equation} \label{upper bound min distance}
d_{R,1}(C) \leq \min \{ d_{R,1}(C_1^\prime), d_{R,1}(C_2^\prime), \ldots, d_{R,1}(C_l^\prime) \}.
\end{equation}

The previous theorem also gives the following corollary for cartesian products:

\begin{corollary} \label{exact grw product}
If $ C = C_1 \times C_2 \times \cdots \times C_l $ and $ 1 \leq r \leq k $, with notation as before, then
\begin{equation} 
\begin{split}
d_{R,r}(C) = \min \{ & d_{R,r_1}(C_1) + d_{R,r_2}(C_2) + \cdots + d_{R,r_l}(C_l) \\
 & \mid r = r_1 + r_2 + \cdots + r_l \}.
\end{split}
\end{equation} \label{grw product}
\end{corollary}

Now we illustrate Theorem \ref{grw bounds} with the following example that includes the MRD $ \mathbb{F}_{q^m} $-linear codes in Subsection \ref{existing code constructions}, item 2, for $ l = 2 $:

\begin{example} \label{example computation}
With notation as in Theorem \ref{grw bounds}, assume that $ l = 2 $, $ n_1, n_2 \leq m $, $ k_1 \leq k_2 $ and take $ C_1 $ and $ C_2 $ as MRD codes (the matrix $ G_{1,2} $ can be arbitrary). In particular, $ d_{R,r_i}(C_i) = n_i - k_i + r_i $ \cite{rgrw} as in Subsection \ref{existing code constructions}, $ 1 \leq r_i \leq k_i $, $ i = 1,2 $. We estimate $ d_{R,r}(C) $ considering three cases:
\begin{enumerate}
\item
Assume $ 1 \leq r \leq k_1 $: The bounds (\ref{lower bound}) and (\ref{upper bound}) give
$$ \min \{ n_1 - k_1, n_2 - k_2 \} + r \leq d_{R,r}(C) \leq n_2 - k_2 + r. $$
\item
Assume $ k_1 < r \leq k_2 $ (if $ k_1 < k_2 $): In this case, in both bounds in Theorem \ref{grw bounds}, it is necessary that $ r_2 > 0 $. Hence, these bounds coincide and give the value $ d_{R,r}(C) = n_2 - k_2 + r $.
\item
Assume $ k_2 < r \leq k $: As in the previous case, now it is necessary that $ r_1 > 0 $ and $ r_2 > 0 $, and thus Theorem \ref{grw bounds} gives the value $ d_{R,r}(C) = n - k +r $, which is optimal by the Singleton bound \cite[Proposition 1]{rgrw}.
\end{enumerate}
\end{example}

Finally, it is natural to ask whether different reductions (see Definition \ref{def reducible}) may give different bounds in Theorem \ref{grw bounds}. In Appendix \ref{app 1}, we show that all reductions have the same main components, thus (\ref{lower bound}) remains unchanged. We now show that (\ref{upper bound min distance}) can always be attained by some particular reduction. Other cases where (\ref{upper bound}) may be attained by some reduction are open.

\begin{proposition} \label{proposition exact for one red}
With notation as in Subsection \ref{subsec reducible}, there exists a reduction $ \overline{\mathcal{R}} = (\overline{G}_{i,j})_{1 \leq i \leq l}^{i \leq j \leq l} $ of $ C $ such that the bound (\ref{upper bound min distance}) is an equality.
\end{proposition}
\begin{proof}
Assume that the minimum rank distance is attained by $ {\rm wt_R}(\mathbf{c}) = d_{R,1}(C) $, for $ \mathbf{c} \in C $. It holds that $ \mathbf{c} = \mathbf{c}_1^\prime + \mathbf{c}_2^\prime + \cdots + \mathbf{c}_l^\prime $, with $ \mathbf{c}_i^\prime \in C_i^\prime $, and $ \mathbf{c}_i^\prime = \mathbf{x}_i G_{i,i}^\prime $ (recall (\ref{row components})), for some $ \mathbf{x}_i \in \mathbb{F}_{q^m}^{k_i} $ and all $ i = 1,2, \ldots, l $. 

We may assume without loss of generality that $ \mathbf{x}_1 \neq \mathbf{0} $. We just need to define $ \overline{G}_{i,i} = G_{i,i} $ and choose matrices $ A_{1,j} \in \mathbb{F}_{q^m}^{k_1 \times k_j} $ and $ \overline{G}_{i,j} \in \mathbb{F}_{q^m}^{k_i \times n_j} $, for $ 1 \leq i \leq l-1 $ and $ i+1 \leq j \leq l $, such that the $ k \times k $ matrix
\begin{displaymath}
A = \left(
\begin{array}{cccccc}
I & A_{1,2} & A_{1,3} & \ldots & A_{1, l-1} & A_{1,l} \\
0 & I & 0 & \ldots & 0 & 0 \\
0 & 0 & I & \ldots & 0 & 0 \\
\vdots & \vdots & \vdots & \ddots & \vdots & \vdots \\
0 & 0 & 0 & \ldots & I & 0 \\
0 & 0 & 0 & \ldots & 0 & I \\
\end{array} \right)
\end{displaymath}
satisfies that $ G = A \overline{G} $, where $ \overline{G} $ is the generator matrix of $ C $ corresponding to $ \overline{\mathcal{R}} = (\overline{G}_{i,j})_{1 \leq i \leq l}^{i \leq j \leq l} $, and
$$ \mathbf{x}_1 A_{1,j} = - \mathbf{x}_j, $$
for $ j = 2,3, \ldots, l $. It is possible to choose such matrices $ A_{1,j} $ because $ \mathbf{x}_1 \neq \mathbf{0} $. Then $ \mathbf{c} = (\mathbf{x} A) \overline{G} $ lies in the first row component of the reduction $ \overline{\mathcal{R}} $ and hence $ d_{R,1}(C) = {\rm wt_R}(\mathbf{c}) \geq d_{R,1}(\overline{C}_1^\prime) $, implying the result.
\end{proof}

We conclude the section with the proof of Theorem \ref{grw bounds}. We need the following lemma:

\begin{lemma} \label{lemma D generators}
With notation as in Subsection \ref{subsec reducible}, define the sets
$$ A_i = \{ 0 \}^{n_1} \times \cdots \times \{ 0 \}^{n_{i-1}} \times (\mathbb{F}_{q^m}^{n_i} \setminus \{ \mathbf{0} \}) \times \mathbb{F}_{q^m}^{n_{i+1}} \times \cdots \times \mathbb{F}_{q^m}^{n_l}, $$
for $ i = 1,2, \ldots, l $. For an $ \mathbb{F}_{q^m} $-linear vector space $ D \subseteq \mathbb{F}_{q^m}^n $, there exist subspaces $ D_i^\prime \subseteq \langle D \cap A_i \rangle $, for $ i $ satisfying $ D \cap A_i \neq \varnothing $, such that $ D = \bigoplus_{D \cap A_i \neq \varnothing} D^\prime_i $ and $ D^\prime_i \cap A_j = \varnothing $ for $ j > i $. 
\end{lemma}
\begin{proof}
We may prove it by induction on the number of indices $ i $ such that $ D \cap A_i \neq \varnothing $. If such number is $ 1 $, the result is trivial by taking $ D^\prime_i = D $, since $ D = \langle D \cap A_i \rangle $. 

Assume that it is larger than $ 1 $ and $ i $ is the smallest index such that $ D \cap A_i \neq \varnothing $. Define $ \widetilde{D} = \sum_{j = i+1}^l \langle D \cap A_j \rangle \neq \{ \mathbf{0} \} $, and let $ D^\prime_i \neq \{ \mathbf{0} \} $ by one of its complementaries in $ D $. It follows that $ D^\prime_i \subseteq \langle D \cap A_i \rangle $ and $ D^\prime_i \cap A_j = \varnothing $, for $ j > i $.

Now, by induction hypothesis, $ \widetilde{D} $ has a decomposition as in the theorem, which together with $ D^\prime_i $ gives the desired decomposition of $ D $.
\end{proof}

\begin{proof} [Proof of Theorem \ref{grw bounds}]
We first prove (\ref{lower bound}). Take an $ r $-dimensional $ \mathbb{F}_{q^m} $-linear subspace $ D \subseteq C $. With notation as in Lemma \ref{lemma D generators}, define $ D_i \subseteq C_i $ as the projection of $ D^\prime_i $ onto the $ i $-th main component, for $ i $ such that $ D \cap A_i \neq \varnothing $. We see that $ \dim(D_i) = \dim(D^\prime_i) $, since $ D_i^\prime \subseteq \langle D \cap A_i \rangle $ and $ D^\prime_i \cap A_j = \varnothing $ for $ j > i $, and by collecting the preimages in $ D^* $ by the projection map of bases of $ D_i^* $, for $ i $ such that $ D \cap A_i \neq \varnothing $, we see that
$$ {\rm wt_R}(D) \geq \sum_{D \cap A_i \neq \varnothing} {\rm wt_R}(D_i), $$
and the result follows by Lemma \ref{equivalent definition slides}.

To prove (\ref{upper bound}), take a decomposition $ r = r_1 + r_2 + \cdots + r_l $, with $ 0 \leq r_i \leq k_i $, for $ i = 1,2, \ldots, l $, and take $ \mathbb{F}_{q^m} $-linear subspaces $ D_i \subseteq C_i^\prime $ with $ \dim(D_i) = r_i $ and $ {\rm wt_R}(D_i) = d_{R,r_i}(C_i^\prime) $. Then define the $ \mathbb{F}_{q^m} $-linear subspace $ D = D_1 \oplus D_2 \oplus \cdots \oplus D_l \subseteq C $, which satisfies $ \dim(D) = r $. By definition, it holds that $ D^* = D_1^* + D_2^* + \cdots + D_l^* $. Hence
\begin{equation*} 
\begin{split}
{\rm wt_R}(D) & \leq {\rm wt_R}(D_1) + {\rm wt_R}(D_2) + \cdots + {\rm wt_R}(D_l) \\
 & = d_{R,r_1}(C_1^\prime) + d_{R,r_2}(C_2^\prime) + \cdots + d_{R,r_l}(C_l^\prime),
\end{split}
\end{equation*} 
and the result follows again by Lemma \ref{equivalent definition slides}.
\end{proof}

\begin{remark}
Observe that the bound (\ref{upper bound}) is valid with the same proof for a general $ \mathbb{F}_{q^m} $-linear code that can be decomposed as a direct sum of $ \mathbb{F}_{q^m} $-linear subcodes $ C = C_1^\prime \oplus C_2^\prime \oplus \cdots \oplus C_l^\prime $.
\end{remark}

\begin{remark} \label{remark exact grw product}
In the general setting of Theorem \ref{grw bounds}, the same result as in Corollary \ref{exact grw product} holds whenever $ C_i $ and $ C_i^\prime $ are rank equivalent (see Section \ref{sec all optimal}), for each $ i = 1,2, \ldots, l $, since in that case it holds that $ d_{R,r}(C_i) = d_{R,r}(C_i^\prime) $ for all $ i =1,2, \ldots, l $ and all $ r = 1,2, \ldots, k_i $.
\end{remark}

\section{MRD reducible codes with MRD main components, and their GRWs}

Among all GRWs, the first weight (the minimum rank distance) is of special importance, as explained at the end of Subsection \ref{subsec GRWs and leakage}. Therefore, it is of interest to study the GRWs of a family of MRD codes, that is, codes that are already optimal for the first weight.

In this section, we find new parameters for which reducible codes are MRD or close to MRD when $ n > m $, extending the family of MRD codes in \cite{reducible} (see Subsection \ref{existing code constructions}), and then give bounds on their GRWs and exact values in the cartesian product case, using the results in the previous section. Hence we give for the first time estimates and exact values of the GRWs of a family of MRD codes with $ n > m $. We will also compare the performance of these codes with those $ \mathbb{F}_q $-linear MRD codes obtained by transposing the matrix representations of codewords in a Gabidulin code \cite[Sec. IV-B]{allertonversion}.

\subsection{Definition of the codes} \label{sec MRD reducible}

Assume $ n > m $ and fix an integer $ 1 \leq k \leq n $. In view of the bound (\ref{bound min distance}), we will consider a reducible code $ C_{red} \subseteq \mathbb{F}_{q^m}^n $ of dimension $ k $ whose main components $ C_1, C_2, \ldots, C_l $ (with notation as in Subsection \ref{subsec reducible}) have as similar parameters as possible. This will allow to obtain reducible codes with minimum rank distance as large as allowed by (\ref{singleton bound}).

First we need the following parameters:
\begin{enumerate}
\item
There exist unique $ l > 0 $ and $ 0 \leq t \leq m-1 $ such that 
$$ n = lm - t. $$
\item
There exist unique $ k^\prime > 0 $ and $ 0 \leq s \leq l-1 $ such that 
$$ k = lk^\prime - s. $$
\item
Define then 
$$ a = \left\lceil \frac{km}{n} \right\rceil - k^\prime, \quad \textrm{and} \quad b = \left\lceil \frac{t}{l} \right\rceil - 1. $$
\item
Finally, define 
$$ t^\prime = l(m-b) - n, $$ 
which satisfies $ 0 < t^\prime \leq l $.
\end{enumerate}

We need the next inequalities to define the desired codes:

\begin{lemma}
It holds that $ k^\prime \leq m-b $ if $ b \geq 0 $, and $ k^\prime \leq m $ if $ b=-1 $.
\end{lemma}
\begin{proof}
For $ b = -1 $, we have that $ t = 0 $ and $ k = lk^\prime - s \leq n = lm $ implies that $ k^\prime \leq m + s/l $. Since $ s < l $, the result holds in this case.

Now assume that $ b \geq 0 $. We have that $ k+s \leq n+l $. Writing $ k $ and $ n $ as above, this inequality reads 
$$ (lk^\prime - s) + s \leq (lm - t) +l, $$
that is, $ lk^\prime + t \leq l(m+1) $ and, dividing by $ l $, it is equivalent to 
$$ k^\prime + \frac{t}{l} - 1 \leq m. $$
The result follows by the definition of $ b $.
\end{proof}

Finally, we give the construction, distinguishing three cases:

\begin{definition} \label{definition tilde code}
Define the reducible code $ C_{red} \subseteq \mathbb{F}_{q^m}^n $ of dimension $ k $ with MRD main components $ C_1, C_2, \ldots, C_l $ as follows:
\begin{enumerate}
\item
If $ t = 0 $ (i.e. $ b = -1 $): Choose $ C_1, C_2, \ldots, C_l $ such that $ l-s $ of them have length $ m $ and dimension $ k^\prime $, and $ s $ of them have length $ m $ and dimension $ k^\prime -1 $. By (\ref{bound min distance}), we have that
$$ d_{R,1}(C_{red}) \geq m - k^\prime + 1. $$
\item
If $ t > 0 $ and $ t^\prime \leq s $: Choose $ C_1, C_2, \ldots, C_l $ such that $ l-s $ of them have length $ m-b $ and dimension $ k^\prime $, $ s-t^\prime $ of them have length $ m-b $ and dimension $ k^\prime -1 $, and $ t^\prime $ of them have length $ m-b-1 $ and dimension $ k^\prime -1 $. By (\ref{bound min distance}), we have that
$$ d_{R,1}(C_{red}) \geq m - b - k^\prime + 1. $$
\item
If $ t > 0 $ and $ t^\prime > s $: Choose $ C_1, C_2, \ldots, C_l $ such that $ l-t^\prime $ of them have length $ m-b $ and dimension $ k^\prime $, $ t^\prime-s $ of them have length $ m-b-1 $ and dimension $ k^\prime $, and $ s $ of them have length $ m-b-1 $ and dimension $ k^\prime -1 $. By (\ref{bound min distance}), we have that
$$ d_{R,1}(C_{red}) \geq m - b - k^\prime. $$
\end{enumerate}
\end{definition}

The next theorem is the first main result of this section, and it gives families of parameters $ m $, $ n $ and $ k $ such that $ C_{red} $ is MRD or almost MRD:

\begin{theorem} \label{new found MRD}
Assume that $ 0 \leq t \leq l $ or $ n \geq m^2 $. The following holds:
\begin{enumerate}
\item
If $ t \leq s $ or $ tk^\prime > ms $, then
$$ d_{R,1}(C_{red}) = \left\lfloor \frac{m}{n}(n-k) + 1 \right\rfloor, $$
attaining (\ref{singleton bound}) if $ n $ divides $ mk $.
\item
If $ t > s $ and $ tk^\prime \leq ms $, then 
$$ d_{R,1}(C_{red}) \geq \left\lfloor \frac{m}{n}(n-k) \right\rfloor. $$
\end{enumerate}
\end{theorem}
\begin{proof}
First we see that we only need to assume $ 0 \leq t \leq l $. Assume that $ n \geq m^2 $. Since $ n = lm - t \geq m^2 $ and $ t \geq 0 $, it holds that $ l \geq m $. Therefore $ t \leq m-1 \leq l-1 $. 

Next we observe that
\begin{equation}
\left\lfloor \frac{m}{n} \left(n - k \right) + 1 \right\rfloor = m - a - k^\prime + 1.
\label{eq first observe}
\end{equation}

Before considering the different cases, we will see that $ a \geq 0 $, and $ a = 0 $ if and only if $ k^\prime t \leq sm $. 

First it holds that $ -1 < km/n - k^\prime $ if and only if
$$ (k^\prime - 1)n < km. $$
Using that $ n = lm - t $ and $ k = lk^\prime -s $, and rearranging terms, this inequality reads
$$ sm + (k^\prime - 1)t < lm + n, $$
which is always true since $ s < l $ and $ k^\prime t \leq k \leq n $. Hence $ a \geq 0 $. On the other hand, $ km/n - k^\prime \leq 0 $ if and only if
$$ nk^\prime \geq km. $$
Using again that $ n = lm - t $ and $ k = lk^\prime -s $, and rearranging terms, this inequality reads $ k^\prime t \leq sm $. This is then the case when $ a = 0 $.

Now we prove item 1 in the theorem:

Assume first that $ t=0 $, then $ d_{R,1}(C_{red}) \geq m - k^\prime + 1 $ and $ a = 0 $, hence the result follows in this case by (\ref{eq first observe}).

Now assume that $ 0 < t \leq s $. Then $ d_{R,1}(C_{red}) \geq m - k^\prime + 1 $ (since $ b=0 $) and $ k^\prime t \leq sm $ holds, since $ k^\prime \leq m $. Then $ a = 0 $ and the result follows in this case by (\ref{eq first observe}).

Next assume that $ tk^\prime > ms $. Then we know that $ a \geq 1 $ and 
$$ \left\lfloor \frac{m}{n} \left(n - k \right) + 1 \right\rfloor \leq m - k^\prime. $$
Since $ b=0 $, we know that $ d_{R,1}(C_{red}) \geq m-k^\prime $, hence the result follows in this case by (\ref{eq first observe}).

Finally, we prove item 2: 

Assume that $ t > s $ and $ tk^\prime \leq ms $. Then we know that $ a = b = 0 $ and $ d_{R,1}(C_{red}) \geq m - b -k^\prime $. Therefore the result follows also in this case by (\ref{eq first observe}) and we are done. 
\end{proof}

\begin{remark}
Observe that the MRD reducible codes in Subsection \ref{existing code constructions}, item 2, are the subfamily of the codes $ C_{red} $ obtained by choosing $ t=s=0 $, and hence are particular cases of the codes in the previous theorem.
\end{remark}

\begin{remark}
Observe that the conditions $ 0 \leq t \leq l $ and $ n \geq m^2 $ only depend on $ m $ and $ n $, but not on $ k $. Hence, for the previous families of values of $ n $ and $ m $, we have obtained MRD or almost MRD codes for all dimensions.
\end{remark}

\begin{remark}
In general, the difference $ b-a $ will be big if $ t $ is much bigger than $ l $. As $ n $ grows, the fact $ t > l $ happens for fewer values of $ t $. Hence the codes $ C_{red} $ are far from optimal when $ n $ is small compared to $ m $ (still $ n > m $) and $ t $ is much bigger than $ l $. 
\end{remark}

\subsection{Estimates and exact values of their GRWs} \label{estimates on their}

The next theorem is the second main result in this section, and it gives estimates of the GRWs of the MRD (or almost MRD) reducible codes $ C_{red} $ from Theorem \ref{new found MRD}, using the lower bound (\ref{lower bound}). 

\begin{theorem} \label{estimates on their theorem}
Let the parameters be as in Theorem \ref{new found MRD}.

Assume first that $ t \leq s $.
\begin{enumerate}
\item
If $ 1 \leq j \leq l- s $ and $ (j-1)k^\prime < r \leq jk^\prime $, or if $ l-s < j \leq l-s+t $ and $ (j-1)(k^\prime-1)+l-s < r \leq j(k^\prime - 1) + l - s $, then 
$$ d_{R,r}(C_{red}) \geq j(m - k^\prime) + r. $$
\item
If $ l - s +t < j \leq l $ and $ (j-1)(k^\prime-1)+l-s < r \leq j(k^\prime - 1) + l - s $, then
$$ d_{R,r}(C_{red}) \geq j(m - k^\prime) + r + (j-l+s-t). $$
\end{enumerate}

Assume now that $ t > s $. 
\begin{enumerate}
\item
If $ 1 \leq j \leq t-s $ and $ (j-1)k^\prime < r \leq jk^\prime $, then
$$ d_{R,r}(C_{red}) \geq j(m - k^\prime-1) + r. $$
\item
If $ t-s < j \leq l-s $ and $ (j-1)k^\prime < r \leq jk^\prime $, or if $ l-s < j \leq l $ and $ (j-1)(k^\prime - 1) +l-s < r \leq j(k^\prime - 1) + l -s $, then
$$ d_{R,r}(C_{red}) \geq j(m - k^\prime) + r -t+s. $$
\end{enumerate}
These cases cover all $ r = 1,2, \ldots, k $ and moreover, if $ C_{red} $ is the cartesian product of its main components $ C_1, C_2, \ldots, C_l $, then all the previous lower bounds are equalities. 
\end{theorem}
\begin{proof}
The result follows from Theorem \ref{grw bounds}. To see it, we just have to use that $ d_{R,r_i}(C_i) = n_i - k_i + r_i $ and see in which way we have to choose $ r_i = 0 $ or $ r_i > 0 $ to obtain the minimum in the bound (\ref{lower bound}), for $ i = 1,2, \ldots, l $. This is a straightforward extension of the calculations in Example \ref{example computation}.
\end{proof}

\subsection{Comparison with other MRD codes} \label{subsec comparison with other MRD}

In this subsection, we will compare the codes $ C_{red} \subseteq \mathbb{F}_{q^m}^n $ from Definition \ref{definition tilde code} with the $ \mathbb{F}_q $-linear MRD codes $ C_{Gab}^T \subseteq \mathbb{F}_{q^m}^n $ obtained by transposing the matrix represenations (see (\ref{eq def matrix representation})) of the codewords in a given $ \mathbb{F}_{q^n} $-linear Gabidulin code $ C_{Gab} \subseteq \mathbb{F}_{q^n}^m $ (see Subsection \ref{existing code constructions}), when $ n > m $. 

The codes $ C_{Gab}^T $ were obtained previously by Delsarte \cite[Th. 6]{delsartebilinear} and have been recently considered for universal secure linear network coding in \cite[Sec. IV-B]{allertonversion}.

We next argue the advantages of the codes $ C_{red} $ over the codes $ C_{Gab}^T $:

\begin{enumerate}
\item
\textit{Generalized rank weights}: Although GRWs have recently been extended to $ \mathbb{F}_q $-linear codes \cite{ravagnaniweights, allertonversion} and its connection to worst-case information leakage has been obtained \cite[Th. 3]{allertonversion}, little is known about them for codes that are not linear over $ \mathbb{F}_{q^m} $. In particular, the GRWs of the codes $ C_{Gab}^T $ are not known yet, except for their minimum rank distance.
\item
\textit{Encoding and decoding complexity}: The complexity of coset encoding and decoding with an $ \mathbb{F}_{q^m} $-linear code, as in Definition \ref{definition ozarow}, is equivalent to the complexity of encoding with one of its generator matrices. 

If $ k_{red} $ denotes the dimension of $ C_{red} $ over $ \mathbb{F}_q $, then the complexity of encoding with a generator matrix coming from one of its reductions is $ O(k_{red} m^2) $ operations over $ \mathbb{F}_{q^m} $, whereas if $ k_{Gab} $ denotes the dimension of $ C_{Gab} $ over $ \mathbb{F}_q $, then the complexity of encoding with one of its generator matrices is $ O(k_{Gab} n^2) $ operations over $ \mathbb{F}_{q^n} $. Therefore it is a higher complexity since $ n > m $, and the difference between both complexities becomes higher the bigger $ n $ is with respect to $ m $.
\item
\textit{Possible parameters obtained}: Since the codes $ C_{Gab}^T $ are obtained from $ \mathbb{F}_{q^n} $-linear codes, their sizes are of the form $ q^N $, where $ N $ is some multiple of $ n $, whereas the sizes of the codes $ C_{red} $ are of the form $ q^M $, where $ M $ is some multiple of $ m $.

Since we are assuming $ n > m $, in a given interval of positive integers, there are more possible parameters attained by the codes $ C_{red} $ than by the codes $ C_{Gab}^T $.
\item
\textit{Stronger security}: The information leakage for a given number of wire-tapped links when using the codes $ C_{red} $ is often much less than the worst case given by their GRWs, as we will see in Section \ref{stronger section}. In particular, looking at their first GRW, we will see that more links can be wire-tapped and still guarantee zero information leakage when using $ C_{red} $ than when using $ C_{Gab}^T $.
\end{enumerate}

\section{All $ \mathbb{F}_{q^m} $-linear codes with optimal GRWs for all fixed packet and code sizes} \label{sec all optimal}

In this section, we obtain all $ \mathbb{F}_{q^m} $-linear codes whose GRWs are all optimal for fixed packet and code sizes ($ m $ and $ k $, respectively), but varying length, $ n $, up to rank equivalence. These codes are particular cases of the codes $ C_{red} $ in the previous section. 

\begin{definition} \label{def optimal all grws}
For fixed $ k $ and $ m $, and for a basis $ \alpha_1, \alpha_2, \ldots, \alpha_m $ of $ \mathbb{F}_{q^m} $ over $ \mathbb{F}_q $, define the $ \mathbb{F}_{q^m} $-linear code $ C_{opt} = C_1 \times C_2 \times \cdots \times C_k \subseteq \mathbb{F}_{q^m}^{km} $, where all $ C_i $ are equal and generated by the vector $ (\alpha_1, \alpha_2, \ldots, \alpha_m) \in \mathbb{F}_{q^m}^m $.
\end{definition}

To claim the above mentioned optimality of these codes, we need the following bounds given in \cite[Lemma 6]{similarities}:

\begin{lemma}[\textbf{\cite{similarities}}] \label{mono}
Given an $ \mathbb{F}_{q^m} $-linear code $ C \subseteq \mathbb{F}_{q^m}^n $ of dimension $ k $, for each $ r = 1,2, \ldots, k-1 $, it holds that
\begin{equation} \label{eq general upper bound monotonicity}
1 \leq d_{R,r+1}(C) - d_{R,r}(C) \leq m. 
\end{equation}
As a consequence, for each $ r = 1,2, \ldots, k $, it holds that 
\begin{equation} \label{eq general upper bound}
d_{R,r}(C) \leq rm.
\end{equation}
\end{lemma}

Observe that these bounds only depend on the packet and code sizes ($ m $ and $ k $, respectively), and they do not depend on the length $ n $.

We first show that the codes $ C_{opt} $ attain the previous bounds, and then prove that they are the only ones with this property:

\begin{proposition} \label{product gabidulin one dime}
Let $ C_{opt} \subseteq \mathbb{F}_{q^m}^{km} $ be the $ \mathbb{F}_{q^m} $-linear code in Definition \ref{def optimal all grws} for given $ k $ and $ m $. Then $ \dim(C_{opt}) = k $ and $ d_{R,r}(C_{opt}) = rm $, for $ r = 1,2, \ldots, k $.
\end{proposition}
\begin{proof}
It holds that $ d_{R,1}(C_i) = m $, for $ i = 1,2, \ldots, k $, since these codes are one-dimensional Gabidulin codes in $ \mathbb{F}_{q^m}^m $ (see Subsection \ref{existing code constructions}). Hence, by Corollary \ref{exact grw product}, we have that
$$ d_{R,k}(C_{opt}) = \sum_{i=1}^k d_{R,1}(C_i) = km. $$
By (\ref{eq general upper bound monotonicity}), it holds that $ d_{R,r}(C_{opt}) = rm $, for $ r = 1,2, \ldots, k $.
\end{proof}

We will use the definition of rank equivalences from \cite[Def. 8]{similarities}, which are stronger than vector space isomorphisms that preserve rank weights:

\begin{definition} [\textbf{\cite{similarities}}]
If $ V \subseteq \mathbb{F}_{q^m}^n $ and $ V^\prime \subseteq \mathbb{F}_{q^m}^{n^\prime} $ are $ \mathbb{F}_{q^m} $-linear Galois closed spaces, we say that a map $ \phi : V \longrightarrow V^\prime $ is a rank equivalence if it is a vector space isomorphism and $ {\rm wt_R}(\phi(\mathbf{c})) = {\rm wt_R}(\mathbf{c}) $, for all $ \mathbf{c} \in V $. 

We say that two codes $ C $ and $ C^\prime $ are rank equivalent if there exists a rank equivalence between $ \mathbb{F}_{q^m} $-linear Galois closed spaces $ V $ and $ V^\prime $ that contain $ C $ and $ C^\prime $, respectively, and mapping bijectively $ C $ to $ C^\prime $.
\end{definition}

Finally, we show that the codes $ C_{opt} $ are the only $ \mathbb{F}_{q^m} $-linear codes attaining (\ref{eq general upper bound}) for fixed packet and code sizes up to rank equivalence:

\begin{theorem} \label{uniqueness}
Let $ C \subseteq \mathbb{F}_{q^m}^n $ be an $ \mathbb{F}_{q^m} $-linear code of dimension $ k $ such that $ d_{R,r}(C) = rm $, for every $ r = 1,2, \ldots, k $. 

Then, for every basis $ \alpha_1, \alpha_2, \ldots, \alpha_m $ of $ \mathbb{F}_{q^m} $ over $ \mathbb{F}_q $, the code $ C $ is rank equivalent to the code $ C_{opt} \subseteq \mathbb{F}_{q^m}^{km} $ in Definition \ref{def optimal all grws}. Moreover, the rank equivalence can be explicitly constructed in polynomial time from any basis of $ C $.
\end{theorem}

We need some preliminary lemmas to prove this result. We start by the following characterization of rank equivalences, which is a particular case of \cite[Th. 5]{similarities}:

\begin{lemma}[\textbf{\cite{similarities}}] \label{rank equivalences}
Let $ \phi : V \longrightarrow V^\prime $ be an $ \mathbb{F}_{q^m} $-linear vector space isomorphism, where $ V \subseteq \mathbb{F}_{q^m}^n $ and $ V^\prime \subseteq \mathbb{F}_{q^m}^{n^\prime} $ are $ \mathbb{F}_{q^m} $-linear Galois closed spaces.

It is a rank equivalence if and only if there exist bases $ \mathbf{v}_1, \mathbf{v}_2, \ldots, \mathbf{v}_t \in \mathbb{F}_q^n $ and $ \mathbf{w}_1, \mathbf{w}_2, \ldots, \mathbf{w}_t \in \mathbb{F}_q^{n^\prime} $ of $ V $ and $ V^\prime $, respectively, and a non-zero element $ \beta \in \mathbb{F}_{q^m} $, such that $ \phi (\mathbf{v}_i) = \beta \mathbf{w}_i $, for $ i = 1,2, \ldots, t $.
\end{lemma}

We now introduce some notation. For a given vector $ \mathbf{c} = (c_1, c_2, \ldots, c_n) \in \mathbb{F}_{q^m}^n $, define $ \mathbf{c}^{[i]} = (c_1^{[i]}, c_2^{[i]}, \ldots, c_n^{[i]}) $, for all integers $ i \geq 0 $. Then define the trace map $ {\rm Tr} : \mathbb{F}_{q^m}^n \longrightarrow \mathbb{F}_q^n $ of the extension $ \mathbb{F}_q \subseteq \mathbb{F}_{q^m} $ as follows
$$ {\rm Tr}(\mathbf{c}) = \mathbf{c} + \mathbf{c}^{[1]} + \mathbf{c}^{[2]} + \cdots + \mathbf{c}^{[m-1]}, $$
for all $ \mathbf{c} \in \mathbb{F}_{q^m}^n $. We have the following two lemmas:

\begin{lemma} \label{matrix gabidulin}
For a basis $ \alpha_1, \alpha_2, \ldots, \alpha_m $ of $ \mathbb{F}_{q^m} $ over $ \mathbb{F}_q $, the matrix $ A = (\alpha_i^{[j-1]})_{1 \leq i,j \leq m} $ over $ \mathbb{F}_{q^m} $ is invertible.
\end{lemma}
\begin{proof}
Well-known. See for instance \cite{gabidulin}.
\end{proof}

\begin{lemma} \label{two bases}
For a basis $ \alpha_1, \alpha_2, \ldots, \alpha_m $ of $ \mathbb{F}_{q^m} $ over $ \mathbb{F}_q $ and the matrix $ A = (\alpha_i^{[j-1]})_{1 \leq i,j \leq m} $, define 
$$ (\beta_1, \beta_2, \ldots, \beta_m) = \mathbf{e}_1 A^{-1}, $$
where $ \mathbf{e}_1 \in \mathbb{F}_{q^m}^m $ is the first vector in the canonical basis. Then $ \beta_1, \beta_2, \ldots, \beta_m $ is also a basis of $ \mathbb{F}_{q^m} $ over $ \mathbb{F}_q $.

Moreover, if $ B = (\beta_i^{[j-1]})_{1 \leq i,j \leq m} $, then 
$$ (\alpha_1, \alpha_2, \ldots, \alpha_m) = \mathbf{e}_1 B^{-1}. $$
\end{lemma}
\begin{proof}
Write $ \boldsymbol\beta = (\beta_1, \beta_2, \ldots, \beta_m) $. Then $ \boldsymbol\beta A = \mathbf{e}_1 $, which means that $ \sum_{i=1}^m \beta_i \alpha_i^{[j-1]} = \delta_{j,1} $, where $ \delta $ is the Kronecker delta. By raising this equation to the power $ [l-1] = q^{l-1} $ and using that $ \delta_{j,l} $ is $ 0 $ or $ 1 $, we see that $ \sum_{i=1}^m \beta_i^{[l-1]} \alpha_i^{[j-1]} = \delta_{j,l} $, that is, $ \boldsymbol\beta^{[l-1]} A = \mathbf{e}_l $, for $ l = 1,2, \ldots, m $.

Let $ \boldsymbol\lambda \in \mathbb{F}_q^m $ be such that $ \boldsymbol\lambda \cdot \boldsymbol\beta = 0 $. By raising this equation to the power $ [l-1] $, for $ l = 1,2, \ldots, m $, we see that $ \boldsymbol\lambda \cdot \boldsymbol\beta^{[l-1]} = 0 $ or, equivalently, $ \boldsymbol\lambda \cdot (\mathbf{e}_l A^{-1}) = 0 $, since $ \boldsymbol\lambda \in \mathbb{F}_q^m $.

Write $ \boldsymbol\mu = (\mu_1, \mu_2, \ldots, \mu_m) = \boldsymbol\lambda (A^{-1})^T $. It holds that
$$ 0 = \boldsymbol\lambda \cdot (\mathbf{e}_l A^{-1}) = (\boldsymbol\lambda (A^{-1})^T) \cdot \mathbf{e}_l = \boldsymbol\mu \cdot \mathbf{e}_l = \mu_l, $$
for $ l = 1,2, \ldots, m $. Therefore, $ \boldsymbol\mu = \mathbf{0} $, thus $ \boldsymbol\lambda = \mathbf{0} $. Hence the elements $ \beta_1, \beta_2, \ldots, \beta_m $ are linearly independent over $ \mathbb{F}_q $.

Finally, since $ \sum_{i=1}^m \beta_i^{[l-1]} \alpha_i^{[j-1]} = \delta_{j,l} $, it holds that $ \sum_{i=1}^m \alpha_i \beta_i^{[j-1]} = \delta_{1,j} = \delta_{j,1} $, which means that $ (\alpha_1, \alpha_2, \ldots, \alpha_m)B = \mathbf{e}_1 $, and we are done.
\end{proof}

We may now prove Theorem \ref{uniqueness}:

\begin{proof}[Proof of Theorem \ref{uniqueness}]
Choose any basis $ \mathbf{b}_1, \mathbf{b}_2, \ldots, \mathbf{b}_k $ of $ C $. Since $ \dim (C^*) = km $ and $ C^* $ is generated by the elements $ \mathbf{b}_s^{[j-1]} $, for $ s = 1,2, \ldots, k $ and $ j = 1,2, \ldots, m $, it follows that these elements are linearly independent over $ \mathbb{F}_{q^m} $.

Define the vector $ \boldsymbol\beta = (\beta_1, \beta_2, \ldots, \beta_m) = \mathbf{e}_1 A^{-1} $, with notation as in the previous lemma. By that lemma, $ \beta_1, \beta_2, \ldots, \beta_m $ constitute a basis of $ \mathbb{F}_{q^m} $ over $ \mathbb{F}_q $, and $ (\alpha_1, \alpha_2, \ldots, \alpha_m) = \mathbf{e}_1 B^{-1} $.

Consider the vectors $ \mathbf{v}_{s,i} = {\rm Tr}(\beta_i \mathbf{b}_s) \in \mathbb{F}_q^n $, for $ s = 1,2, \ldots, k $ and $ i = 1,2, \ldots, m $. Assume that there exist $ \lambda_{s,i} \in \mathbb{F}_q $ such that $ \sum_{s=1}^k \sum_{i=1}^m \lambda_{s,i} \mathbf{v}_{s,i} = \mathbf{0} $. Then it holds that
$$ \sum_{j=1}^{m} \sum_{s=1}^k \left( \sum_{i=1}^m \lambda_{s,i} \beta_i^{[j-1]} \right) \mathbf{b}_s^{[j-1]} = \mathbf{0}. $$
Hence $ \sum_{i=1}^m \lambda_{s,i} \beta_i^{[j-1]} = 0 $, for $ s = 1,2, \ldots, k $ and $ j = 1,2, \ldots, m $, which implies that $ \lambda_{s,i} = 0 $, for $ s = 1,2, \ldots, k $ and $ i = 1,2, \ldots, m $.

Therefore, the elements $ \mathbf{v}_{s,i} $, for $ s = 1,2, \ldots, k $ and $ i = 1,2, \ldots, m $, constitute a basis of $ C^* $ and are vectors in $ \mathbb{F}_q^n $. Now define the $ \mathbb{F}_{q^m} $-linear vector space isomorphism $ \psi : C^* \longrightarrow \mathbb{F}_{q^m}^{km} $ by $ \psi(\mathbf{v}_{s,i}) = \mathbf{e}_{(s-1)m + i} $, for $ s = 1,2, \ldots, k $ and $ i = 1,2, \ldots, m $. By Lemma \ref{rank equivalences}, $ \psi $ is a rank equivalence and, moreover,
$$ \mathbf{b}_s = \sum_{j=1}^m \sum_{i=1}^m \alpha_i \beta_i^{[j-1]} \mathbf{b}_s^{[j-1]} = \sum_{i=1}^m \alpha_i {\rm Tr}(\beta_i \mathbf{b}_s) = \sum_{i=1}^m \alpha_i \mathbf{v}_{s,i}. $$
It follows that $ \mathbf{v}_s = \psi (\mathbf{b}_s) = \sum_{i=1}^m \alpha_i \mathbf{e}_{(s-1)m + i} $, and the vectors $ \mathbf{v}_s $, for $ s = 1,2, \ldots, k $, constitute a basis of $ \psi(C) $. Finally, this means that $ \psi(C) = C_{opt} $ and we are done.
\end{proof}

\begin{remark}
As explained in Subsection \ref{subsec universal secure}, given an $ \mathbb{F}_{q^m} $-linear code $ C \subseteq \mathbb{F}_{q^m}^n $ of dimension $ k $, the parameter $ m $ represents the packet length, $ k $ represents the number of linearly independent packets that we may send using $ C $, or its size, and $ n $ represents the number of outgoing links from the source.

Due to the bounds (\ref{eq general upper bound}), if $ m $ and $ k $ are fixed and $ n $ is not restricted, then the code $ C_{opt} $ is the only $ \mathbb{F}_{q^m} $-linear code whose GRWs are all optimal, and hence is the only $ \mathbb{F}_{q^m} $-linear optimal code regarding information leakage in the network, up to rank equivalence.
\end{remark}

%
%
%
%

\begin{remark}
The codes $ C_{opt} \subseteq \mathbb{F}_{q^m}^{km} $ do not only have optimal GRWs, but the difference between two consecutive weights is the largest possible by (\ref{eq general upper bound monotonicity}):
$$ d_{R,r+1}(C_{opt}) = d_{R,r}(C_{opt}) + m, $$
for $ r = 1,2, \ldots, k-1 $. However, for a Gabidulin code $ C_{Gab} $ as in Subsection \ref{existing code constructions}, the difference between two consecutive weights is the smallest possible by (\ref{eq general upper bound monotonicity}):
$$ d_{R,r+1}(C_{Gab}) = d_{R,r}(C_{Gab}) + 1, $$
for $ r = 1,2, \ldots, k-1 $. 

Therefore, when using $ C_{opt} $, an adversary that obtains $ r $ packets of information, by listening to the smallest possible number of links, needs to listen to at least $ m $ more links in order to obtain one more packet of information. However, when using $ C_{Gab} $, the adversary only needs to listen to one more link to obtain one more packet of information. 
\end{remark}

\section{Stronger security of reducible codes} \label{stronger section}

On the error correction side, it is well-known that reducible codes can correct a substantial amount of rank errors beyond half of their minimum rank distance \cite[Sec. III.A]{reducible}.

The aim of this section is to show that, on the security side, when using a reducible code $ C $, an eavesdropper may in many cases obtain less than $ r $ packets of information even if he or she wire-taps at least $ d_{R,r}(C) $ links in the network (see Subsection \ref{subsec GRWs and leakage}). 

Setting $ r=1 $ and using an MRD reducible code (as in Section \ref{sec MRD reducible}), this means that the eavesdropper obtains no information even when wire-tapping strictly more links than those allowed by other MRD codes ($ \mathbb{F}_{q^m} $-linear or $ \mathbb{F}_q $-linear), by \cite[Th. 3]{allertonversion}. 

The above mentioned stronger security is obtained by upper bounding the dimensions of the code intersected with Galois closed spaces, due to Equation (\ref{leaked info}). We explain this in the remarks at the end of the section.

The following is the main result of this section, where we denote by $ \pi_i : \mathbb{F}_{q^m}^n \longrightarrow \mathbb{F}_{q^m}^{n_i} $ the projection map onto the coordinates corresponding to the $ i $-th main component $ C_i \subseteq \mathbb{F}_{q^m}^{n_i} $, for $ i = 1,2, \ldots, l $, with notation as in Subsection \ref{subsec reducible}.

\begin{theorem} \label{theorem main stronger security}
Let $ V \subseteq \mathbb{F}_{q^m}^n $ be an $ \mathbb{F}_{q^m} $-linear Galois closed space and assume that, for each $ i = 1,2, \ldots, l $, there exists $ 0 \leq r_i \leq k_i $ such that $ \dim(\pi_i(V)) \leq d_{R,r_i}(C_i) $, with notation as in Subsection \ref{subsec reducible}. Then 
$$ \dim(C \cap V) \leq \left( \sum_{i=1}^l r_i \right) - \# \{ i \mid \dim(\pi_i(V)) < d_{R,r_i}(C_i) \}. $$
In particular, if $ \dim(\pi_i(V)) < d_{R,1}(C_i) $, for $ i = 1,2, \ldots, l $, then
$$ \dim(C \cap V) = 0. $$
\end{theorem}

Before proving this theorem, we give two consequences of interest. In the first, we give a sufficient condition for the eavesdropper to obtain less than $ r $ packets of information, for a given $ r $, as in the second paragraph of this section:

\begin{corollary} \label{corollary stronger security 2}
Let the notation be as in Subsection \ref{subsec reducible}, let $ 1 \leq r \leq k $ and let $ V \subseteq \mathbb{F}_{q^m}^n $ be an $ \mathbb{F}_{q^m} $-linear Galois closed space. Assume that $ r = \sum_{i=1}^l r_i $, where $ 1 \leq r_i \leq k_i $ and $ \dim(\pi_i(V)) \leq d_{R,r_i}(C_i) $, for $ i = 1,2, \ldots, l $, and for some $ j $ it holds that $ \dim(\pi_j(V)) < d_{R,r_j}(C_j) $. Then 
$$ \dim(C \cap V) < r. $$
\end{corollary}

The second consequence is just the previous theorem applied to the codes in Definition \ref{def optimal all grws}:

\begin{corollary} \label{corollary stronger security 1}
Let $ C_{opt} \subseteq \mathbb{F}_{q^m}^{km} $ be the code in Definition \ref{def optimal all grws}, and let $ V \subseteq \mathbb{F}_{q^m}^{km} $ be an $ \mathbb{F}_{q^m} $-linear Galois closed space. Then
$$ \dim(C_{opt} \cap V) \leq \# \{ i \mid \pi_i(V) = \mathbb{F}_{q^m}^m \}. $$
\end{corollary}

Finally, we prove Theorem \ref{theorem main stronger security}. We need the following lemma:

\begin{lemma}
Let $ V \subseteq \mathbb{F}_{q^m}^n $ be an $ \mathbb{F}_{q^m} $-linear Galois closed space, and let the notation be as in Subsection \ref{subsec reducible}. It holds that
$$ \dim(C \cap V) \leq \sum_{i=1}^l \dim(C_i \cap \pi_i(V)). $$
\end{lemma} 
\begin{proof}
Let $ D = C \cap V \subseteq C $ and let the notation be as in Lemma \ref{lemma D generators}. Since $ D = \bigoplus_{D \cap A_i \neq \varnothing} D^\prime_i $, we just need to show that $ \dim(D^\prime_i) \leq \dim(C_i \cap \pi_i(V)) $, for $ i $ such that $ D \cap A_i \neq \varnothing $.

Fix such an index $ i $, and let $ \rho_i : D^\prime_i \longrightarrow C_i \cap \pi_i(V) $ be the restriction of $ \pi_i $ to $ D^\prime_i $. It is well-defined since $ \pi_i(D^\prime_i) \subseteq \pi_i(V) $ by definition of $ D $, and $ \pi_i(D^\prime_i) \subseteq C_i $ since $ D^\prime_i \subseteq \langle C \cap A_i \rangle $. 

Finally, we see that $ \rho_i $ is one to one since $ D^\prime_i \subseteq \langle C \cap A_i \rangle $ and $ D^\prime_i \cap A_j = \varnothing $ for $ j > i $, and we are done. 
\end{proof}

\begin{proof} [Proof of Theorem \ref{theorem main stronger security}]
First observe that $ \pi_i(V) \subseteq \mathbb{F}_{q^m}^{n_i} $ is again Galois closed, for $ i = 1,2, \ldots, l $. By definition of GRWs, if $ \dim(\pi_i(V)) < d_{R,r_i}(C_i) $, then $ \dim(C_i \cap \pi_i(V)) < r_i $, for $ i $ such that $ r_i > 0 $. On the other hand, if $ \dim(\pi_i(V)) \leq d_{R,r_i}(C_i) $ and $ r_i < k_i $, then by monotonicity of GRWs \cite[Lemma 4]{rgrw}, it holds that $ \dim(\pi_i(V)) < d_{R,r_i+1}(C_i) $, which implies that $ \dim(C_i \cap \pi_i(V)) < r_i+1 $, that is, $ \dim(C_i \cap \pi_i(V)) \leq r_i $. Finally, if $ \dim(\pi_i(V)) \leq d_{R,k_i}(C_i) $, then it is trivial that $ \dim(C_i \cap \pi_i(V)) \leq \dim(C_i) = k_i $. 

The result follows then from the previous lemma.
\end{proof}

\begin{remark}
In the situation of Corollary \ref{corollary stronger security 2}, if $ \dim(\pi_i(V)) \leq d_{R,r_i}(C_i) $, for $ i = 1,2, \ldots, l $ and with strict inequality for some $ j $, then an eavesdropper that obtains $ \mathbf{c} B^T $, where $ B $ generates $ V $, gains less than $ r $ packets of information about the original packets by Equation (\ref{leaked info}). 

Observe that the previous condition implies that $ \dim(V) < \sum_{i=1}^l d_{R,r_i}(C_i) $. We know from the bound (\ref{lower bound}) that if $ \dim(V) < \sum_{i=1}^l d_{R,s_i}(C_i) $ for all possible decompositions $ r= \sum_{i=1}^l s_i $, then $ \dim(C \cap V ) < r $. 

However, many $ \mathbb{F}_{q^m} $-linear Galois closed spaces may satisfy $ \dim(\pi_i(V)) < d_{R,r_i}(C_i) $, for $ i = 1,2, \ldots, l $, and a given decomposition $ r = \sum_{i=1}^l r_i $, but may also satisfy $ \dim(V) \geq \sum_{i=1}^l d_{R,s_i}(C_i) $ for some other decomposition $ r = \sum_{i=1}^l s_i $.

Take for instance $ V = V_1 \times V_2 \times \cdots \times V_l $, where $ V_i \subseteq \mathbb{F}_{q^m}^{n_i} $ are $ \mathbb{F}_{q^m} $-linear Galois closed spaces satisfying $ \dim(V_i) \leq d_{R,r_i}(C_i) $, for $ i = 1,2, \ldots, l $ and with strict inequality for some $ j $, but $ \dim(V) = \sum_{i=1}^l \dim(V_i) \geq d_{R,r}(C) $. 
\end{remark}

\begin{remark}
In the particular case of Corollary \ref{corollary stronger security 1}, to obtain at least $ r $ packets of information, it must hold that $ \pi_i(V) $ is the whole space $ \mathbb{F}_{q^m}^m $ for at least $ r $ indices $ i $. Take for instance $ V= V_1 \times V_2 \times \cdots \times V_k $, where $ V_i \subsetneq \mathbb{F}_{q^m}^n $ satisfies $ \dim(V_i) = m-1 $, for $ i = 1,2, \ldots, k $. In that case, $ \dim(V) = k(m-1) $, which is usually much bigger than $ d_{R,1}(C) = m $. However, the adversary still obtains no information about the original packets.
\end{remark}

\section{Related properties of reducible codes}

In this section, we study some secondary properties of reducible codes that are related to their GRWs.

\subsection{Cartesian product conditions}

In this subsection, we gather sufficient and necessary conditions for reducible codes to be rank equivalent to cartesian products (see Section \ref{sec all optimal} for the definition of rank equivalence).

We start by using Galois closures and generalized rank weights to see whether an $ \mathbb{F}_{q^m} $-linear code that can be decomposed as a direct sum of smaller codes is rank equivalent to the cartesian product of these codes. It can be seen as a converse statement to Corollary \ref{exact grw product}.

\begin{proposition} \label{char product}
Given an $ \mathbb{F}_{q^m} $-linear code $ C = C_1^\prime \oplus C_2^\prime \oplus \cdots \oplus C_l^\prime \subseteq \mathbb{F}_{q^m}^n $, with $ k_i = \dim(C_i^\prime) $, for $ i = 1,2, \ldots, l $, and $ k = \dim(C) $, we have that $ C^* = C_1^{\prime *} + C_2^{\prime *} + \cdots + C_l^{\prime *} $ and the following conditions are equivalent:
\begin{enumerate}
\item
$ C $ is rank equivalent to a cartesian product $ C_1 \times C_2 \times \cdots \times C_l \subseteq \mathbb{F}_{q^m}^n $, where $ C_i \subseteq \mathbb{F}_{q^m}^{n_i} $ is rank equivalent to $ C_i^\prime $, and the equivalence map from $ C $ to the product is the product of the equivalence maps from $ C_i^\prime $ to $ C_i $.
\item
$ C^* = C_1^{\prime *} \oplus C_2^{\prime *} \oplus \cdots \oplus C_l^{\prime *} $.
\item 
$ d_{R,k}(C) = d_{R,k_1}(C_1^\prime) + d_{R,k_2}(C_2^\prime) + \cdots + d_{R,k_l}(C_l^\prime) $.
\item
For all $ r = 1,2, \ldots, k $, it holds that
\begin{equation*} 
\begin{split}
d_{R,r}(C) = \min \{ & d_{R,r_1}(C_1^\prime) + d_{R,r_2}(C_2^\prime) + \cdots + d_{R,r_l}(C_l^\prime) \\
 & \mid r = r_1 + r_2 + \cdots + r_l, 0 \leq r_i \leq k_i \}. \\
\end{split}
\end{equation*}
\end{enumerate}
\end{proposition}
\begin{proof}
It is trivial that item 1 implies item 4 by Corollary \ref{exact grw product}. It is also trivial that item 4 implies item 3, and items 2 and 3 are equivalent since $ d_{R,k}(C) = \dim(C^*) $ and $ d_{R,k_i}(C_i^\prime) = \dim(C_i^{\prime *}) $, for $ i = 1,2, \ldots, l $, by Lemma \ref{equivalent definition slides}.

Now we prove that item 2 implies item 1. Define $ V_i = C_i^{\prime *} $, for $ i =1,2, \ldots, l $, and $ V = C^* $. We may assume that $ C $ is not rank degenerate, that is, $ V = \mathbb{F}_{q^m}^n $. Therefore, $ n = \dim(V) $, $ n_i = \dim(V_i) $, for $ i = 1,2, \ldots, l $, and $ n = n_1 + n_2 + \cdots + n_l $.

On the other hand, define a vector space isomorphisms $ \psi_i : V_i \longrightarrow \mathbb{F}_{q^m}^{n_i} $, for $ i = 1,2, \ldots, l $, by sending a basis of $ V_i $ of vectors in $ \mathbb{F}_q^n $ to the canonical basis of $ \mathbb{F}_{q^m}^{n_i} $. It is a rank equivalence by Lemma \ref{rank equivalences}. Define $ C_i = \psi_i(C_i^\prime) $. Therefore, $ C_i $ and $ C_i^\prime $ are rank equivalent by definition.

Finally, define $ \psi : V = V_1 \oplus V_2 \oplus \cdots \oplus V_l \longrightarrow \mathbb{F}_{q^m}^n $ by 
$$ \psi (\mathbf{c}_1 + \mathbf{c}_2 + \cdots + \mathbf{c}_l) = (\psi_1(\mathbf{c}_1), \psi_2(\mathbf{c}_2), \ldots, \psi_l(\mathbf{c}_l)), $$
where $ \mathbf{c}_i \in V_i $, for all $ i = 1,2, \ldots, l $. It holds that $ \psi $ maps vectors in $ \mathbb{F}_q^n $ to vectors in $ \mathbb{F}_q^n $ and is a vector space isomorphism. Hence, it is a rank equivalence by Lemma \ref{rank equivalences} and verifies the required conditions. 
\end{proof}

\begin{corollary}
With notation as in Subsection \ref{subsec reducible}, if $ C_i $ is rank equivalent to $ C_i^\prime $, for all $ i = 1,2, \ldots, l $, then $ C $ is rank equivalent to $ C_1 \times C_2 \times \cdots \times C_l $.
\end{corollary}

Observe that the previous corollary states that Remark \ref{remark exact grw product} is actually implied by Corollary \ref{exact grw product}.

On the other hand, we may use the column components to see wether $ C = C_1 \times C_2 \times \cdots \times C_l $ exactly. The proof is straightforward:

\begin{proposition}
With notation as in Subsection \ref{subsec reducible}, the following conditions are equivalent:
\begin{enumerate}
\item
$ C = C_1 \times C_2 \times \cdots \times C_l $.
\item
$ C = \widehat{C} $.
\item
$ k_i = \widehat{k}_i $, for all $ i = 1,2, \ldots, l $.
\item
For each $ j = 2,3, \ldots, l $, the rows in $ G_{i,j} $, $ 1 \leq i \leq j-1 $, are contained in the main component $ C_j $.
\end{enumerate}
\end{proposition}

\subsection{Rank degenerate conditions}

Recall the definition of rank degenerate codes from \cite[Def. 9]{similarities}:

\begin{definition} [\textbf{\cite{similarities}}] \label{definition degenerate}
An $ \mathbb{F}_{q^m} $-linear code $ C \subseteq \mathbb{F}_{q^m}^n $ of dimension $ k $ is rank degenerate if $ d_{R,k}(C) < n $.
\end{definition}

In network coding, a code is rank degenerate if it can be applied to a network with strictly less outgoing links from the source node (see \cite{slides, similarities} for more details).

In this subsection, we study sufficient and necessary conditions for reducible codes to be rank degenerate.

\begin{proposition}
With notation as in Subsection \ref{subsec reducible}, it holds that:
\begin{enumerate}
\item
If $ C $ is rank degenerate, then there exists an $ 1 \leq i \leq l $ such that $ C_i $ is rank degenerate.
\item
If there exists an $ 1 \leq j \leq l $ such that $ \widehat{C}_j $ is rank degenerate, then $ C $ is rank degenerate.
\end{enumerate}
\end{proposition}
\begin{proof}
We prove each item separately:
\begin{enumerate}
\item
It follows from 
$$ d_{R,k}(C) \geq d_{R,k_1}(C_1) + d_{R,k_2}(C_2) + \cdots + d_{R,k_l}(C_l), $$
which follows from Theorem \ref{grw bounds}, and the fact that $ C $ has length $ n $ and $ C_i $ has length $ n_i $, for $ i = 1,2, \ldots, l $.
\item
We have that $ C \subseteq \widehat{C} $. Hence $ C^* \subseteq \widehat{C}^* $ and
$$ d_{R,k}(C) = \dim(C^*) \leq \dim(\widehat{C}^*) = d_{R,\widehat{k}}(\widehat{C}), $$
by Lemma \ref{equivalent definition slides}, and 
$$ d_{R,\widehat{k}}(\widehat{C}) = d_{R,\widehat{k}_1}(\widehat{C}_1) + d_{R,\widehat{k}_2}(\widehat{C}_2) + \cdots + d_{R,\widehat{k}_l}(\widehat{C}_l), $$
by Corollary \ref{exact grw product}, hence the item follows, using now that $ \widehat{C}_j $ has length $ n_j $, for $ j = 1,2, \ldots, l $.
\end{enumerate}
\end{proof}


\begin{corollary}
If $ C = C_1 \times C_2 \times \cdots \times C_l $, then $ C $ is rank degenerate if and only if there exists an $ 1 \leq i \leq l $ such that $ C_i $ is rank degenerate.
\end{corollary}

\subsection{Duality and bounds on GRWs} 

With notation as in Subsection \ref{subsec reducible}, it is shown in \cite{reducible} that the dual of the reducible code $ C $ has a generator matrix of the form
\begin{displaymath}
H = \left(
\begin{array}{cccccc}
H_{1,1} & 0 & 0 & \ldots & 0 & 0 \\
H_{2,1} & H_{2,2} & 0 & \ldots & 0 & 0 \\
H_{3,1} & H_{3,2} & H_{3,3} & \ldots & 0 & 0 \\
\vdots & \vdots & \vdots & \ddots & \vdots & \vdots \\
H_{l-1,1} & H_{l-1,2} & H_{l-1,3} & \ldots & H_{l-1,l-1} & 0 \\
H_{l,1} & H_{l,2} & H_{l,3} & \ldots & H_{l,l-1} & H_{l,l} \\
\end{array} \right),
\end{displaymath}
where $ H_{i,i} $ is a generator matrix of $ C_i^\perp $, for $ i = 1,2, \ldots, l $. 

We see that reversing the order of the row blocks does not change the code, and reversing the order of the column blocks gives a rank equivalent code. Hence, denoting by $ (C^\perp)_i^\prime $ the subcode of $ C^\perp $ generated by the matrix 
$$ H_i^\prime = (H_{i,1}, \ldots, H_{i,i-1}, H_{i,i}, 0, \ldots, 0), $$
for $ i = 1,2, \ldots, l $, we may obtain analogous bounds on the generalized rank weights of $ C^\perp $ to those in Theorem \ref{grw bounds}. We leave the details to the reader.

An upper bound on the GRW of $ C^\perp $ using column components of $ C $ that follows from Corollary \ref{exact grw product} is the following:

\begin{proposition}
With notation as in Subsection \ref{subsec reducible}, it holds that
\begin{equation} \label{lower bound duals}
\begin{split}
d_{R,r}(C^\perp) \leq \min \{ & d_{R,\widehat{r}_1}(\widehat{C}^\perp_1) + d_{R,\widehat{r}_2}(\widehat{C}^\perp_2) + \cdots + d_{R,\widehat{r}_l}(\widehat{C}^\perp_l) \\
 & \mid r = \widehat{r}_1 + \widehat{r}_2 + \cdots + \widehat{r}_l, 0 \leq \widehat{r}_i \leq \widehat{k}_i \},
\end{split}
\end{equation}
for $ r = 1,2, \ldots, n- \widehat{k} $ (observe that $ n- \widehat{k} \leq n-k $).
\end{proposition}
\begin{proof}
It holds that $ C \subseteq \widehat{C} $, hence $ \widehat{C}^\perp \subseteq C^\perp $, and the result follows then from Corollary \ref{exact grw product} and the fact that $ \widehat{C}^\perp = \widehat{C}_1^\perp \times \widehat{C}_2^\perp \times \cdots \times \widehat{C}_l^\perp $.
\end{proof}

In particular, if $ \widehat{k} < n $, it holds that
\begin{equation} \label{lower bound duals min}
d_{R,1}(C^\perp) \leq \min \{ d_{R,1}(\widehat{C}^\perp_1), d_{R,1}(\widehat{C}^\perp_2), \ldots, d_{R,1}(\widehat{C}^\perp_l) \}.
\end{equation}

\subsection{MRD rank}

Recall from \cite[Prop. 1]{rgrw} the (classical) Singleton bound on GRWs:
\begin{equation} \label{singleton}
d_{R,r}(C) \leq n - k + r,
\end{equation}
for any $ \mathbb{F}_{q^m} $-linear code $ C \subseteq \mathbb{F}_{q^m}^n $, where $ k = \dim(C) $ and $ 1 \leq r \leq k $. By monotonicity of GRWs \cite[Lemma 4]{rgrw}, if the $ r $-th weight of $ C $ attains the Singleton bound, then the $ s $-th weight of $ C $ also attains it, for all $ s \geq r $. The minimum of such $ r $ is called the MRD rank of the code \cite[Def. 1]{jerome}:

\begin{definition}[\textbf{\cite{jerome}}]
For an $ \mathbb{F}_{q^m} $-linear code $ C \subseteq \mathbb{F}_{q^m}^n $ of dimension $ k $, we define its MRD rank as the minimum positive integer $ r $ such that $ d_{R,r}(C) = n - k + r $, and denote it by $ r(C) $. 

If $ d_{R,k}(C) < n $, then we define $ r(C) = k + 1 $. 
\end{definition}

Observe that the last part of the previous definition is a redefinition of rank degenerate codes. We have the next characterization of $ r(C) $ given in \cite[Cor. III.3]{jerome}:

\begin{lemma}[\textbf{\cite{jerome}}]
For an $ \mathbb{F}_{q^m} $-linear code $ C \subseteq \mathbb{F}_{q^m}^n $ of dimension $ k $, it holds that
$$ r(C) = k - d_{R,1}(C^\perp) + 2, $$
defining $ d_{R,1}(\{ \mathbf{0} \}) = n+1 $ for the case $ C = \mathbb{F}_{q^m}^n $.
\end{lemma}

In particular, from the bounds obtained so far, we derive the following result on the MRD rank of a reducible code:

\begin{proposition}
Let the notation be as in Subsection \ref{subsec reducible}. It holds that
\begin{equation} \label{lower bound MRD}
k - r(C) \geq \min \{ k_1 - r(C_1), k_2 - r(C_2), \ldots, k_l - r(C_l) \}
\end{equation} 
and
\begin{equation} \label{upper bound MRD 1}
k - r(C) \leq \min \{ \widehat{k}_1 - r(\widehat{C}_1), \widehat{k}_2 - r(\widehat{C}_2), \ldots, \widehat{k}_l - r(\widehat{C}_l) \}.
\end{equation} 
Moreover, denote by $ k_{i,j} $ and $ r_{i,j} $ the dimension and MRD rank of the $ \mathbb{F}_{q^m} $-linear code with parity check matrix $ H_{i,j} $, respectively, with notation as in the previous subsection, for $ i = 2,3, \ldots, l $ and $ j = 1,2, \ldots, i-1 $. Then 
\begin{equation} \label{upper bound MRD 2}
\begin{split}
k - r(C) \leq \min \{ & k_i - r(C_i) + \sum_{H_{i,j} \neq 0} (k_{i,j} - r_{i,j} +2) \\ 
 & \mid i = 1,2, \ldots, l \}.
\end{split}
\end{equation}
\end{proposition}
\begin{proof}
The bound (\ref{lower bound MRD}) follows from the previous lemma and the bound (\ref{bound min distance}). The bound (\ref{upper bound MRD 1}) follows from the previous lemma and the bound (\ref{lower bound duals min}).

Now we prove the bound (\ref{upper bound MRD 2}). From the previous lemma and the bound (\ref{upper bound min distance}), we obtain that
$$ k - r(C) \leq \min \{ d_{R,1}((C^\perp)^\prime_1, (C^\perp)^\prime_2, \ldots, (C^\perp)^\prime_l) \}, $$
with notation as in the previous subsection. Now, if $ d_{i,j} $ denotes the minimum rank distance of the $ \mathbb{F}_{q^m} $-linear code with parity check matrix $ H_{i,j} $, it follows that
$$ d_{R,1}((C^\perp)_i^\prime) \leq d_{R,1}(C_i^\perp) + \sum_{H_{i,j} \neq 0} d_{i,j}, $$
and the result follows again from the previous lemma.
\end{proof}

The MRD rank of the code $ C $ in Example \ref{example computation} was obtained directly using Theorem \ref{grw bounds}. However, it could be directly obtained using the previous proposition.

We conclude with the cartesian product case:

\begin{corollary}
With notation as in the previous proposition, if $ C = C_1 \times C_2 \times \cdots \times C_l $, it holds that
$$ k - r(C) = \min \{ k_1 - r(C_1), k_2 - r(C_2), \ldots, k_l - r(C_l) \}, $$
and all the bounds in the previous proposition are equalities.
\end{corollary}

\subsection{Particular constructions}

To conclude, in this subsection we briefly recall some constructions of reducible codes in the literature introduced to improve the minimum Hamming distance of cartesian products of codes, and see when they may give improvements for the rank distance.

Recall the well-known $ (\mathbf{u}, \mathbf{u} + \mathbf{v}) $-construction by Plotkin \cite{plotkin}. Take $ \mathbb{F}_{q^m} $-linear codes $ C_1, C_2 \subseteq \mathbb{F}_{q^m}^n $, and define the $ \mathbb{F}_{q^m} $-linear code $ C \subseteq \mathbb{F}_{q^m}^{2n} $ by
$$ C = \{ (\mathbf{u}, \mathbf{u} + \mathbf{v}) \mid \mathbf{u} \in C_1, \mathbf{v} \in C_2 \}. $$
Denoting by $ d_H(D) $ the minimum Hamming distance of a code $ D $, it holds that $ d_H(C_1 \times C_2) = \min \{ d_H(C_1), d_H(C_2) \} $, whereas $ d_H(C) = \min \{ 2 d_H(C_1), d_H(C_2) \} $, hence improving the minimum Hamming distance of the cartesian product if $ d_H(C_1) < d_H(C_2) $.

Observe that $ C $ is reducible. However, its first row component is obviously rank equivalent to its first main component. By Proposition \ref{char product}, $ C $ and $ C_1 \times C_2 $ are rank equivalent. Hence the $ (\mathbf{u}, \mathbf{u} + \mathbf{v}) $-construction gives nothing but cartesian products for the rank metric.

We may apply the same argument for the so-called matrix-product codes \cite{blackmore}, which are a generalization of the previous construction. Let the notation be as in Subsection \ref{subsec reducible}, fix a non-singular matrix $ A \in \mathbb{F}_{q^m}^{l \times l} $ and assume that $ N = n_1 = n_2 = \ldots = n_l $. Define the $ \mathbb{F}_{q^m} $-linear code $ C = (C_1, C_2, \ldots, C_l) A \subseteq \mathbb{F}_{q^m}^{n} $ with generator matrix 
\begin{displaymath}
G = \left(
\begin{array}{cccc}
a_{1,1} G_1 & a_{1,2} G_1 & \ldots & a_{1,l} G_1 \\
a_{2,1} G_2 & a_{2,2} G_2 & \ldots & a_{2,l} G_2 \\
\vdots & \vdots & \ddots & \vdots \\
a_{l,1} G_l & a_{l,2} G_l & \ldots & a_{l,l} G_l \\
\end{array} \right).
\end{displaymath}
If $ A $ is upper triangular, we see that $ C $ is a reducible code. Just as before, if $ A \in \mathbb{F}_q^{l \times l} $, then $ C $ is rank equivalent to $ C_1 \times C_2 \times \cdots \times C_l $, and thus this construction gives nothing but cartesian products.

In the following examples we see that, as an alternative, the $ (\mathbf{u}, \alpha \mathbf{u} + \mathbf{v}) $-construction, for $ \alpha \in \mathbb{F}_{q^m} \setminus \mathbb{F}_q $, and $ (\mathbf{u}, \mathbf{u}^{[i]} + \mathbf{v}) $-construction, for $ 0 < i < m $, may improve the minimum rank distance of the cartesian product.

\begin{example}
Consider $ \alpha \in \mathbb{F}_{q^m} \setminus \mathbb{F}_q $, $ n = 3 $, $ C_1 \subseteq \mathbb{F}_{q^m}^3 $ generated by $ (1,0,0) $ and $ C_2 \subseteq \mathbb{F}_{q^m}^3 $ generated by $ (0,\alpha, \alpha^{[1]}) $ and $ (0,\alpha^{[1]}, \alpha^{[2]}) $. Let $ C $ be the $ (\mathbf{u}, \alpha \mathbf{u} + \mathbf{v}) $-construction of the codes $ C_1 $ and $ C_2 $.

It holds that $ d_{R,1}(C_1 \times C_2) = 1 $, whereas $ d_{R,1}(C) = 2 $.
\end{example}

\begin{example}
Consider $ \alpha \in \mathbb{F}_{q^m} \setminus \mathbb{F}_q $, $ n = 3 $, $ C_1 \subseteq \mathbb{F}_{q^m}^3 $ generated by $ (\alpha,0,0) $ and $ C_2 \subseteq \mathbb{F}_{q^m}^3 $ generated by $ (0,\alpha, \alpha^{[1]}) $ and $ (0,\alpha^{[1]}, \alpha^{[2]}) $. Let $ C $ be the $ (\mathbf{u}, \mathbf{u}^{[1]} + \mathbf{v}) $-construction of the codes $ C_1 $ and $ C_2 $.

Again, it holds that $ d_{R,1}(C_1 \times C_2) = 1 $, whereas $ d_{R,1}(C) = 2 $.
\end{example}

\section{Conclusion and open problems}

In this paper, we have studied the security performance of reducible codes in network coding when used in the form of coset coding schemes. We have obtained lower bounds on their generalized rank weights (GRWs) that extend the known lower bound on their minimum rank distance \cite{reducible} and which give exact values for cartesian products, and we have obtained upper bounds that are always reached for the minimum rank distance and some reduction. We have obtained maximum rank distance (MRD) reducible codes with MRD main components for new parameters, extending the families of MRD codes for $ n > m $ considered in \cite{reducible} and \cite{silva-universal}. 

We have obtained all $ \mathbb{F}_{q^m} $-linear codes whose GRWs are all optimal, for all fixed packet and code sizes up to rank equivalence. The given code construction is a cartesian product of full-length one-dimensional Gabidulin codes and has the minimum possible length required by the optimality of their GRWs. As we have shown, these codes do not only have optimal GRWs, but the difference between every two consecutive GRWs is the packet lenght, which is optimal, in constrast with Gabidulin codes, for which this difference is the minimum possible. Thus if the length of the code is big enough or not restricted, then the given construction behaves much better than Gabidulin codes in secure network coding.

Afterwards we have shown that, when using reducible codes, a wire-tapping adversary obtains in many cases less information than that described by their GRWs. In particular, when using MRD reducible codes or those with optimal GRWs for fixed packet and code sizes, the eavesdropper obtains no information about the sent packets even when wire-tapping more links than those allowed by other MRD codes.

Finally, we have studied some secondary related properties of reducible codes: Characterizations to be rank equivalent to cartesian products of codes, characterizations to be rank degenerate, bounds on their dual codes, MRD ranks, and alternative constructions to the well-known $ (\mathbf{u}, \mathbf{u} + \mathbf{v}) $-construction.

To conclude, we list a few open problems of interest regarding the security behaviour of reducible codes:
\begin{enumerate}
\item
Find other cases when the bounds in Theorem \ref{grw bounds} are equalities, apart from the cases covered in Corollary \ref{grw product} and Proposition \ref{proposition exact for one red}.
\item
Find new parameters for which reducible codes are MRD, or prove the impossibility that a reducible code is MRD for certain parameters.
\item
Prove or disprove the optimality of the codes in Section \ref{sec all optimal} among $ \mathbb{F}_q $-linear codes. We remark here that no sharp bounds such as those in Lemma \ref{mono} are known for general $ \mathbb{F}_q $-linear codes, to the best of our knowledge.
\end{enumerate}

\appendices

\section{Uniqueness of reductions} \label{app 1}

In this appendix, we discuss the uniqueness of the main components, row components and column components of a reducible code (see Subsection \ref{subsec reducible}). We will show that the main components remain unchanged by changing the reduction or by rank equivalence, hence the bound (\ref{lower bound}) remains unchanged. However, the row components may change by changing the reduction, and the column components may change by a rank equivalence. Hence the bounds (\ref{upper bound}) and (\ref{lower bound duals}) may change in those cases. See Proposition \ref{proposition exact for one red}, for instance.

Fix a reducible code $ C \subseteq \mathbb{F}_{q^m}^n $, with notation as in Subsection \ref{subsec reducible}.

\begin{proposition} \label{uniqueness conditions 1}
Given another reduction $ \widehat{\mathcal{R}} $ of $ C $ with the same row and column block sizes as $ \mathcal{R} $, it holds that the main components and column components of $ \widehat{\mathcal{R}} $ and $ \mathcal{R} $ are the same, respectively.
\end{proposition}
\begin{proof}
Let $ \widehat{\mathcal{R}} = (\widehat{G}_{i,j})_{1 \leq i \leq l}^{i \leq j \leq l} $ and let $ \widehat{G} $ be the generator matrix of $ C $ given by this reduction. Since the matrices $ G_{i,i} $ have full rank, there exist matrices $ A_{i,j} \in \mathbb{F}_{q^m}^{k_i \times k_j} $, for $ i = 1,2, \ldots, l $ and $ j = i, i+1, \ldots, l $, such that the $ k \times k $ matrix
\begin{displaymath}
A = \left(
\begin{array}{cccccc}
A_{1,1} & A_{1,2} & A_{1,3} & \ldots & A_{1, l-1} & A_{1,l} \\
0 & A_{2,2} & A_{2,3} & \ldots & A_{2, l-1} & A_{2,l} \\
0 & 0 & A_{3,3} & \ldots & A_{3, l-1} & A_{3,l} \\
\vdots & \vdots & \vdots & \ddots & \vdots & \vdots \\
0 & 0 & 0 & \ldots & A_{l-1,l-1} & A_{l-1,l} \\
0 & 0 & 0 & \ldots & 0 & A_{l,l} \\
\end{array} \right)
\end{displaymath}
satisfies that $ \widehat{G} = A G $. Then it holds that $ \widehat{G}_{i,i} = A_{i,i} G_{i,i} $, for $ i = 1,2, \ldots, l $, and the main components of both reductions coincide. In addition, it holds that 
\begin{displaymath}
\left(
\begin{array}{c}
 \widehat{G}_{1,j} \\
 \widehat{G}_{2,j} \\
 \vdots \\
 \widehat{G}_{j,j} \\
\end{array} \right) = \left(
\begin{array}{cccc}
A_{1,1} & A_{1,2} & \ldots & A_{1,j} \\
0 & A_{2,2} & \ldots & A_{2,j} \\
\vdots & \vdots & \ddots & \vdots \\
0 & 0 & \ldots & A_{j,j} \\
\end{array} \right)  \left(
\begin{array}{c}
 G_{1,j} \\
 G_{2,j} \\
 \vdots \\
 G_{j,j} \\
\end{array} \right),
\end{displaymath}
and the column components of both reductions also coincide.
\end{proof}

\begin{proposition} \label{uniqueness conditions 2}
Assume that the main components of the reduction $ \mathcal{R} $ of $ C $ are not rank degenerate. Let $ \mathcal{R}^\prime $ be a reduction of an $ \mathbb{F}_{q^m} $-linear code $ C^\prime $ that is rank equivalent to $ C $, with the same row and column block sizes as $ \mathcal{R} $, and such that the rank equivalence maps the rows of the generator matrix corresponding to $ \mathcal{R} $ to the rows of the generator matrix corresponding to $ \mathcal{R}^\prime $. Then the main components and row components of $ \mathcal{R}^\prime $ and $ \mathcal{R} $ are rank equivalent, respectively.
\end{proposition}
\begin{proof}
Let $ \mathcal{R}^\prime = (G^\prime_{i,j})_{1 \leq i \leq l}^{i \leq j \leq l} $ and let $ G^\prime $ be the generator matrix of $ C^\prime $ given by this reduction. By hypothesis and by Lemma \ref{rank equivalences}, we may assume that the rank equivalence is given by $ \phi(\mathbf{c}) = \mathbf{c} A $, for $ \mathbf{c} \in \mathbb{F}_{q^m}^n $, for some $ n \times n $ matrix 
\begin{displaymath}
A = \left(
\begin{array}{cccccc}
A_{1,1} & A_{1,2} & A_{1,3} & \ldots & A_{1, l-1} & A_{1,l} \\
A_{2,1} & A_{2,2} & A_{2,3} & \ldots & A_{2, l-1} & A_{2,l} \\
A_{3,1} & A_{3,2} & A_{3,3} & \ldots & A_{3, l-1} & A_{3,l} \\
\vdots & \vdots & \vdots & \ddots & \vdots & \vdots \\
A_{l-1,1} & A_{l-1,2} & A_{l-1,3} & \ldots & A_{l-1,l-1} & A_{l-1,l} \\
A_{l,1} & A_{l,2} & A_{l,3} & \ldots & A_{l,l-1} & A_{l,l} \\
\end{array} \right),
\end{displaymath}
with coefficients in $ \mathbb{F}_q $, and such that $ G^\prime = G A $. Looking at the generator matrices of the last row components of $ \mathcal{R} $ and $ \mathcal{R}^\prime $, we see that
$$ (0, \ldots, 0, G_{l,l}^\prime) = (G_{l,l} A_{l,1}, G_{l,l} A_{l,2}, \ldots, G_{l,l} A_{l,l}), $$
which implies that $ G_{l,l} A_{l,j} = 0 $, for $ j = 1,2, \ldots, l-1 $. This means that the columns of $ A_{l,j} $ are in $ C_l^\perp $. However, since their coefficients lie in $ \mathbb{F}_q $, these columns have rank weight equal to $ 1 $. 

On the other hand, we are assuming that the main components of $ \mathcal{R} $ are not rank degenerate, which in particular means that $ d_R(C_l^\perp) > 1 $ (see \cite[Def. 26 and Cor. 28]{slides}). Therefore, all the columns in $ A_{l,j} $ are the zero vector, that is, $ A_{l,j} = 0 $, for $ j = 1,2, \ldots, l-1 $.

If we now look at the generator matrices of the $ (l-1) $-th row components of $ \mathcal{R} $ and $ \mathcal{R}^\prime $, we see that
$$ (0, \ldots, 0, G_{l-1,l-1}^\prime, G_{l-1,l}^\prime) = (G_{l-1,l-1} A_{l-1,1}, \ldots $$
$$ G_{l-1,l-1} A_{l-1,l-1}, G_{l-1,l-1} A_{l-1,l} + G_{l-1,l}A_{l,l}), $$
which implies that $ G_{l-1,l-1} A_{l-1,j} = 0 $, for $ j = 1,2, \ldots, l-2 $. In the same way as before, we see that this implies that $ A_{l-1,j} = 0 $, for $ j = 1,2, \ldots, l-2 $.

Continuing iteratively in this way, we see that $ A_{i,j} = 0 $, for $ i > j $. In other words, we have that $ A $ is again of the form
\begin{displaymath}
A = \left(
\begin{array}{cccccc}
A_{1,1} & A_{1,2} & A_{1,3} & \ldots & A_{1, l-1} & A_{1,l} \\
0 & A_{2,2} & A_{2,3} & \ldots & A_{2, l-1} & A_{2,l} \\
0 & 0 & A_{3,3} & \ldots & A_{3, l-1} & A_{3,l} \\
\vdots & \vdots & \vdots & \ddots & \vdots & \vdots \\
0 & 0 & 0 & \ldots & A_{l-1,l-1} & A_{l-1,l} \\
0 & 0 & 0 & \ldots & 0 & A_{l,l} \\
\end{array} \right).
\end{displaymath}
As in the proof of Proposition \ref{uniqueness conditions 1}, this implies that the main components and row components of $ \mathcal{R} $ and $ \mathcal{R}^\prime $ are rank equivalent, respectively.
\end{proof}



%

\section*{Acknowledgement}

The author wishes to thank Olav Geil and Diego Ruano for fruitful discussions and comments. He also wishes to thank the editor and anonymous reviewers for their very helpful comments.

\ifCLASSOPTIONcaptionsoff
  \newpage
\fi


\begin{thebibliography}{10}
\providecommand{\url}[1]{#1}
\csname url@samestyle\endcsname
\providecommand{\newblock}{\relax}
\providecommand{\bibinfo}[2]{#2}
\providecommand{\BIBentrySTDinterwordspacing}{\spaceskip=0pt\relax}
\providecommand{\BIBentryALTinterwordstretchfactor}{4}
\providecommand{\BIBentryALTinterwordspacing}{\spaceskip=\fontdimen2\font plus
\BIBentryALTinterwordstretchfactor\fontdimen3\font minus
  \fontdimen4\font\relax}
\providecommand{\BIBforeignlanguage}[2]{{%
\expandafter\ifx\csname l@#1\endcsname\relax
\typeout{** WARNING: IEEEtranS.bst: No hyphenation pattern has been}%
\typeout{** loaded for the language `#1'. Using the pattern for}%
\typeout{** the default language instead.}%
\else
\language=\csname l@#1\endcsname
\fi
#2}}
\providecommand{\BIBdecl}{\relax}
\BIBdecl

\bibitem{ahlswede}
R.~Ahlswede, N.~Cai, S.-Y.~R. Li, and R.~W. Yeung, ``Network information
  flow,'' \emph{IEEE Trans. Inform. Theory}, vol.~46, no.~4, pp. 1204--1216,
  Jul 2000.

\bibitem{blackmore}
T.~D. Blackmore and G.~H. Norton, ``Matrix-product codes over $ \mathbb{F}_q
  $,'' \emph{Applicable Algebra in Engineering, Communications and Computing},
  vol.~12, no.~6, pp. 477--500, 2001.

\bibitem{cai-yeung}
N.~Cai and R.~W. Yeung, ``Network coding and error correction,'' in \emph{Proc.
  2002 IEEE Information Theory Workshop}, 2002, pp. 119--122.

\bibitem{secure-network}
------, ``Secure network coding,'' in \emph{Proc. 2002 IEEE International
  Symposium on Information Theory}, 2002, p. 323.

\bibitem{delsartebilinear}
P.~Delsarte, ``Bilinear forms over a finite field, with applications to coding
  theory,'' \emph{Journal of Combinatorial Theory, Series A}, vol.~25, no.~3,
  pp. 226--241, 1978.

\bibitem{jerome}
J.~Ducoat, ``Generalized rank weights: A duality statement,'' in \emph{Topics
  in Finite Fields}, ser. Comtemporary Mathematics, G.~L.~M. G.~Kyureghyan and
  A.~Pott, Eds.\hskip 1em plus 0.5em minus 0.4em\relax American Mathematical
  Society, 2015, vol. 632, pp. 114--123.

\bibitem{oggiercyclic}
J.~Ducoat and F.~E. Oggier, ``Rank weight hierarchy of some classes of cyclic
  codes,'' in \emph{IEEE Information Theory Workshop (ITW), 2014}, 2014, pp.
  142--146.

\bibitem{gabidulin}
E.~M. Gabidulin, ``Theory of codes with maximum rank distance,'' \emph{Problems
  Inform. Transmission}, vol.~21, no.~1, pp. 1--12, 1985.

\bibitem{reducible}
E.~M. Gabidulin, A.~V. Ourivski, B.~Honary, and B.~Ammar, ``Reducible rank
  codes and their applications to cryptography,'' \emph{IEEE Trans. Inform.
  Theory}, vol.~49, no.~12, pp. 3289--3293, 2003.

\bibitem{ideals}
E.~M. Gabidulin, A.~V. Paramonov, and O.~V. Tretjakov, ``Ideals over a
  non-commutative ring and their application in cryptology,'' in \emph{Advances
  in Cryptology — EUROCRYPT ’91}, ser. Lecture Notes in Computer Science,
  D.~W. Davies, Ed.\hskip 1em plus 0.5em minus 0.4em\relax Springer Berlin
  Heidelberg, 1991, vol. 547, pp. 482--489.

\bibitem{slides}
\BIBentryALTinterwordspacing
R.~Jurrius and R.~Pellikaan, ``On defining generalized rank weights,''
  \emph{Adv. Math. Comm.}, vol.~11, no.~1, pp. 225--235, Feb 2017. 
\BIBentrySTDinterwordspacing

\bibitem{Koetter2003}
R.~K{\"o}tter and M.~Medard, ``An algebraic approach to network coding,''
  \emph{IEEE/ACM Trans. Networking}, vol.~11, no.~5, pp. 782--795, Oct 2003.

\bibitem{rgrw}
J.~Kurihara, R.~Matsumoto, and T.~Uyematsu, ``Relative generalized rank weight
  of linear codes and its applications to network coding,'' \emph{IEEE Trans.
  Inform. Theory}, vol.~61, no.~7, pp. 3912--3936, 2015.

\bibitem{linearnetwork}
S.-Y. Li, R.~Yeung, and N.~Cai, ``Linear network coding,'' \emph{IEEE Trans.
  Inform. Theory}, vol.~49, no.~2, pp. 371--381, Feb 2003.

\bibitem{grwreducible}
U.~Mart{\'i}nez-Pe{\~n}as, ``Generalized rank weights of reducible codes,
  optimal cases and related properties,'' in \emph{2016 IEEE International
  Symposium on Information Theory (ISIT)}, Jul 2016, pp. 1959--1963.

\bibitem{similarities}
------, ``On the similarities between generalized rank and {H}amming weights
  and their applications to network coding,'' \emph{IEEE Trans. Inform.
  Theory}, vol.~62, no.~7, pp. 4081--4095, Jul 2016.

\bibitem{allertonversion}
\BIBentryALTinterwordspacing
U.~Mart{\'i}nez-Pe{\~n}as and R.~Matsumoto, ``Unifying notions of generalized
  weights for universal security on wire-tap networks,'' in \emph{Proceedings
  of the 54th Annual Allerton Conference on Communication, Control, and
  Computing}, Sep 2016, pp. 800--807. 
\BIBentrySTDinterwordspacing

\bibitem{oggier}
F.~E. Oggier and A.~Sboui, ``On the existence of generalized rank weights,'' in
  \emph{Proceedings of the International Symposium on Information Theory and
  its Applications, {ISITA} 2012, Honolulu, HI, USA, October 28-31, 2012},
  2012, pp. 406--410.

\bibitem{ozarow}
L.~H. Ozarow and A.~D. Wyner, ``Wire-tap channel {II},'' in \emph{Advances in
  Cryptology: EUROCRYPT 84}, ser. Lecture Notes in Comput. Sci.\hskip 1em plus
  0.5em minus 0.4em\relax Springer Berlin Heidelberg, 1985, vol. 209, pp.
  33--50.

\bibitem{plotkin}
M.~Plotkin, ``Binary codes with specified minimum distance,'' \emph{IRE Trans.
  Inform. Theory}, vol.~6, no.~4, pp. 445--450, Sep 1960.

\bibitem{ravagnaniweights}
A.~Ravagnani, ``Generalized weights: An anticode approach,'' \emph{Journal of
  Pure and Applied Algebra}, vol. 220, no.~5, pp. 1946--1962, 2016.

\bibitem{on-metrics}
D.~Silva and F.~R. Kschischang, ``On metrics for error correction in network
  coding,'' \emph{IEEE Trans. Inform. Theory}, vol.~55, no.~12, pp. 5479--5490,
  2009.

\bibitem{silva-universal}
------, ``Universal secure network coding via rank-metric codes,'' \emph{IEEE
  Trans. Inform. Theory}, vol.~57, no.~2, pp. 1124--1135, Feb. 2011.

\bibitem{stichtenoth}
H.~Stichtenoth, ``On the dimension of subfield subcodes,'' \emph{IEEE Trans.
  Inform. Theory}, vol.~36, no.~1, pp. 90--93, 1990.

\bibitem{wyner}
A.~D. Wyner, ``The wire-tap channel,'' \emph{The Bell System Technical
  Journal}, vol.~54, no.~8, pp. 1355--1387, Oct. 1975.

\bibitem{zamir}
R.~Zamir, S.~Shamai, and U.~Erez, ``Nested linear/lattice codes for structured
  multiterminal binning,'' \emph{IEEE Trans. Inform. Theory}, vol.~48, no.~6,
  pp. 1250--1276, Jun 2002.

\end{thebibliography}
\end{document}